\documentclass[12pt]{article}

\usepackage{microtype}

\usepackage{ bbold }
\usepackage{graphics}

\usepackage{pst-node}
\usepackage{tikz-cd}
\usepackage{enumitem}
\usepackage{wrapfig,color}

\usepackage{amsmath}
\usepackage{amssymb}
\usepackage{amsthm}
\usepackage{xspace}
\usepackage[normalem]{ulem}
\usepackage{graphicx}
\usepackage{epsfig}
\usepackage{ifpdf}
\usepackage{url}
\usepackage{latexsym}
\usepackage[mathscr]{euscript}
\usepackage{xspace}
\usepackage{color}
\usepackage{makeidx}
\usepackage{wrapfig,color}
\usepackage{stackrel}
\usepackage{algorithm}

\usepackage[all]{xy}

\usepackage{mathtools}

\DeclarePairedDelimiter\set{\{}{\}}

\title{Computing Bottleneck Distance for Multi-parameter Interval Decomposable Persistence Modules}

\author{Tamal K. Dey$^*$~~ Cheng Xin\thanks{Department of Computer Science and Engineering, The Ohio State University. \texttt{dey.8@osu.edu, xin.108@buckeyemail.osu.edu}}}

\date{}

\newtheorem{theorem}{Theorem} 

\newtheorem{remark}[theorem]{Remark}

\newtheorem{corollary}[theorem]{Corollary}
\newtheorem{fact}[theorem]{Fact}

\newtheorem{proposition}[theorem]{Proposition}
\newtheorem{definition}[theorem]{Definition}

\newcommand{\Int}{\mathbb{Z}}
\newcommand{\Real}{\mathbb{R}}
\newcommand{\Rn}{\Real^n}

\newcommand{\Id}{\mathbb{1}}

\newcommand{\VecSp}{\mathbf{Vec}}
\newcommand{\vecsp}{\mathbf{vec}}
\newcommand*{\rom}[1]{\expandafter\@slowromancap\romannumeral #1@}
\newcommand{\Img}{\mathrm{im\,}}
\newcommand{\coim}{\mathrm{coim\,}}
\newcommand{\dm}{\mathrm{dm}}
\newcommand{\dl}{\mathrm{dl}}
\newcommand{\Poset}{\mathbb{P}}
\newcommand{\nmod}{\Rn\mbox{-}\mathbf{mod}}
\newcommand{\RB}{\bar{\mathbb{R}}}
\newcommand{\field}[1] {{\mathbb{#1}}}

\definecolor{darkred}{rgb}{1, 0.1, 0.3}
\definecolor{darkblue}{rgb}{0.1, 0.1, 1}

\begin{document}

\maketitle

\begin{abstract}
	Computation of the interleaving distance between persistence modules is a central task in topological data analysis. For $1$-parameter persistence modules, thanks to the isometry theorem, this can be done by computing the bottleneck distance with known efficient algorithms. The question is open for most $n$-parameter persistence modules, $n>1$, because of the well recognized complications of the indecomposables. Here, we consider a reasonably complicated class called {\em $n$-parameter interval decomposable} modules whose indecomposables may have a description of non-constant complexity. We present a polynomial time algorithm to compute the bottleneck distance for these modules from indecomposables, which bounds the interleaving distance from above, and give another algorithm to compute a new distance called {\em dimension distance} that bounds it from below. An earlier version of this paper considered only the $2$-parameter interval decomposable modules~\cite{DeyCheng18}.
\end{abstract}


\section{Introduction} \label{sec: intro}
\begin{sloppypar}
	Persistence modules have become an important object of study in topological data analysis in that they serve as an intermediate between the raw input data and the output summarization with persistence diagrams. The classical persistence theory 
	\cite{edelsbrunner2010computational} for $\mathbb{R}$-valued functions produces 1-parameter persistence modules, which is a sequence of vector spaces (homology groups with a field coefficient) with linear maps over $\mathbb{R}$ seen as a poset. It is known that~\cite{Crawley-Boevey,Webb}, this sequence can be decomposed uniquely into a set of intervals called {\em bars} which is also represented as points in $\mathbb{R}^2$ called the persistence diagrams~\cite{cohen2007stability}. The space of these diagrams can be equipped with a metric $d_B$ called the {\em bottleneck distance}. Cohen-Steiner et al.~\cite{cohen2007stability} showed that $d_B$ is bounded from above by the input function perturbation measured in infinity norm. Chazal et al.~\cite{Chazal:2009} generalized the result by showing that
	the bottleneck distance is bounded from above by a distance $d_I$ called the {\em interleaving distance} between two persistence modules; see also~\cite{BCM,Bubenik2014,deSilva2016} for further generalizations. Lesnick~\cite{lesnick2015theory} (see also~\cite{Bauer:2014:IMB:2582112.2582168,chazal2012structure}) established the isometry theorem which showed that indeed $d_I=d_B$. Consequently, $d_I$ for 1-parameter persistence modules can be computed exactly by efficient algorithms known for computing $d_B$
	[see e.g.,] \cite{edelsbrunner2010computational,kerber2017geometry}. 
	The status however is not so well settled for
	multi-parameter persistence modules~\cite{Carlsson2009} arising from $\Rn$-valued functions.
\end{sloppypar}

\begin{sloppypar}
	Extending the concept from 1-parameter modules,
	Lesnick~\cite{lesnick2015theory} defined the interleaving distance for $n$-parameter persistence modules, and proved its stability and universality. The definition of the bottleneck distance, however, is not readily extensible mainly because the bars for finitely presented $n$-parameter modules called {\em indecomposables} are far more complicated though are guaranteed to be essentially unique by Krull-Schmidt theorem~\cite{Atiyah1956}. Nonetheless, one can define $d_B$  as the supremum of the pairwise interleaving distances between indecomposables, which in some sense generalizes the concept in 1-parameter due to the isometry theorem.  
	Then, straightforwardly, $d_I\leq d_B$ as observed in~\cite{botnan2016algebraic}, but the converse is not necessarily true. For some special cases, results in the converse direction have started to appear. Botnan and Lesnick~\cite{botnan2016algebraic} proved that, in 2-parameter, $d_B\leq \frac{5}{2} d_I$ for what they called block decomposable modules. Bjerkevic~\cite{bjerkevik2016stability} improved this result to $d_B\leq d_I$. Furthermore, he extended it by proving that $d_B\leq (2n-1)d_I$ for rectangle decomposable $n$-parameter modules and $d_B \leq (n-1)d_I$ for free $n$-parameter modules. He gave an example for exactness of this bound when $n=2$. 
\end{sloppypar}

\begin{sloppypar}
	Unlike 1-parameter modules, the question of estimating $d_I$ for $n$-parameter modules through efficient algorithms is largely open~\cite{Botnan17}.
	Multi-dimensional matching distance  introduced in \cite{cerri2013betti} provides a lower bound to interleaving distance \cite{landi2014rank} and can be approximated within any error threshold by algorithms proposed in \cite{biasotti2011new, cerri2011new}. But, it cannot provide an upper bound like $d_B$.
	For free, block, rectangle, and triangular decomposable modules, one can compute $d_B$ by computing pairwise interleaving distances between indecomposables in constant time because they have a description of constant complexity. 
	Due to the results mentioned earlier, $d_I$ can be estimated within a constant or dimension-dependent factors by computing $d_B$ for these modules.  It is not obvious how to do the same for the larger class of interval decomposable modules  mentioned in the literature~\cite{bjerkevik2016stability,botnan2016algebraic} where indecomposables may not have constant complexity. These are modules whose indecomposables are bounded by "stair-cases". Our main contribution is a polynomial time algorithm that, given indecomposables, computes $d_B$ exactly for $n$-parameter interval decomposable modules. The algorithm draws upon various geometric and algebraic analysis of the interval decomposable modules that may be of independent interest. It is known that no lower bound in terms of $d_B$ for $d_I$ may exist for these modules~\cite{botnan2016algebraic}. To this end, we complement our result by proposing a distance $d_0$ called {\em dimension distance} that is efficiently computable and satisfies the condition $d_0\leq d_I$. An earlier version of this paper considered only the 2-parameter interval decomposable modules~\cite{DeyCheng18}.
\end{sloppypar}



\section{Persistence modules} 
\begin{sloppypar}
	Our goal is to compute the bottleneck distance between two $n$-parameter interval decomposable persistence modules. The bottleneck distance, originally defined for 1-parameter persistence modules \cite{cohen2007stability} (also see \cite{Bauer:2014:IMB:2582112.2582168}), and later extended to multi-parameter persistence modules
	\cite{botnan2016algebraic} is known to bound the interleaving distance between two persistence modules from above.
\end{sloppypar}

Let $\mathbb k$ be a field, $\VecSp$ be the category of vector spaces over $\mathbb k$, and $\vecsp$ be the subcategory of finite dimensional vector spaces. In what follows, for simplicity, we assume $\mathbb{k}=\mathbb{Z}/2\mathbb{Z}$.
\begin{definition}[Persistence module]
	Let $\Poset$ be a poset category. A $\Poset$-indexed persistence module is a functor $M:\Poset\rightarrow \VecSp$. If $M$ takes values in $\vecsp$, we say $M$ is pointwise finite dimensional (p.f.d).
	The $\Poset$-indexed persistence modules themselves form another category where the natural transformations between functors constitute the morphisms. 
\end{definition}

Here we consider the poset category to be $\Rn$ with the standard partial order and all modules to be p.f.d. We call $\Rn$-indexed persistence modules as $n$-parameter persistence modules. The category of $n$-parameter modules is denoted as $\nmod$. For an $n$-parameter module $M\in \nmod $, we use notation $M_x:=M(x)$ and $\rho^M_{x\rightarrow y}:=M(x\leq y)$. 

\begin{definition}[Shift] For any $\delta \in \mathbb{R}$, we denote  $\vec{\delta}=(\delta, \cdots, \delta)=\delta \cdot \vec{e}$, where $\vec{e}=\sum_i e_i$ with $\{e_i\}_{i=1}^n$ being the standard basis of $\mathbb{R}^n$. 
	We define a shift functor $(\cdot)_{\rightarrow \delta}:\nmod \rightarrow \nmod$ where $M_{\rightarrow \delta}:=(\cdot)_{\rightarrow \delta} (M)$ is given by $M_{\rightarrow \delta}(x)=M(x+\vec{\delta})$ 
	and $M_{\rightarrow\delta}(x\leq y)=M(x+\vec{\delta} \leq y+\vec{\delta})$. In other words, $M_{\rightarrow \delta}$ is the module $M$ shifted diagonally by $\vec\delta$.
\end{definition}

The following definition of interleaving taken from ~\cite{oudot2015persistence} adapts the original definition designed for 1-parameter modules in~\cite{chazal2012structure} to $n$-parameter modules.
\begin{definition}[Interleaving]
	For two persistence modules $M$ and $N$, and $\delta \geq 0$, a $\delta$-interleaving between $M$ and $N$ are two families of linear maps $\{\phi_x: M_x \rightarrow N_{x+\vec{\delta}}\}_{x\in \mathbb{R}^n}$ and $\{\psi_x: N_x \rightarrow M_{x+\vec{\delta}}\}_{x\in \mathbb{R}^n}$ satisfying the following two conditions (see Appendix~\ref{app:miss1} for commutative diagrams):
	
	\begin{itemize}
		\item $\forall x \in \mathbb{R}^n, \rho^M_{x \rightarrow x+2\vec{\delta}} = \psi_{x+\vec\delta} \circ \phi_{x}$ and
		$\rho^N_{x \rightarrow x+2\vec{\delta}} = \phi_{x+\vec\delta} \circ \psi_{x}$
		\item $\forall x\leq y \in \mathbb{R}^n, 
		\phi_{y} \circ \rho^M_{x \rightarrow y} =  \rho^N_{x \rightarrow y} \circ \phi_{x} $ and   
		$\psi_{y} \circ \rho^N_{x \rightarrow y} =  \rho^M_{x \rightarrow y} \circ \psi_{x} $ symmetrically
	\end{itemize}

	If such a $\delta$-interleaving exists, we say $M$ and $N$ are $\delta$-interleaved. We call the first condition \textit{triangular commutativity} and the second condition \textit{square commutativity}.
\end{definition}

\begin{definition}[Interleaving distance]
	Define the interleaving distance between modules $M$ and $N$ as  $d_I(M, N)=\inf_{\delta}\{M \mbox{ and } N \mbox{ are } \delta\textit{-interleaved}\}$. 
	We say $M$ and $N$ are $\infty$-interleaved if they are not $\delta$-interleaved for any $\delta\in \Real^+$, and assign $d_I(M, N)=\infty$.
\end{definition}





\begin{definition}[Matching]
	A matching  $\mu: A \nrightarrow B$ between two multisets $A$ and $B$ is a partial bijection, that is, $\mu: A' \rightarrow B'$ for some $A' \subseteq A$ and $B' \subseteq B$. We say $\Img\mu = B', \coim\mu = A'$.
\end{definition}

For the next definition~\cite{botnan2016algebraic}, we call a module \textit{$\delta$-trivial} if $\rho^M_{x \rightarrow x+\vec\delta} = 0$ for all $x\in \Rn$. 

\begin{definition}[Bottleneck distance]
	Let $M \cong \bigoplus_{i=1}^{m} M_{i}$ and $N\cong \bigoplus_{j=1}^{n} N_{j}$ be two persistence modules, where $M_{i}$ and  $N_{j}$ are indecomposable submodules of $M$ and $N$ respectively. Let $I=\{1,\dotsb,m\}$ and $J=\{1,\dotsb,n\}$. We say $M$ and $N$ are $\delta$-matched for $\delta \geq 0$ if there exists a matching
	$\mu:I \nrightarrow J$ so that,
	(i) $i\in I\setminus \coim\mu \implies M_{i} \mbox{ is } 2\delta$-trivial,
	(ii) $j\in   J\setminus \Img\mu \implies N_{j} \mbox{ is } 2\delta$-trivial, and
	(iii) $i\in            \coim\mu \implies M_{i}\ and\ N_{\mu(i)}\ are\ \delta$-interleaved.
	
	The bottleneck distance is defined as
	$$d_B(M, N) = \inf\{\delta\mid M \mbox{ and } N \mbox{ are }\delta \mbox{-matched}\}.$$
\end{definition} 

The following fact observed in~\cite{botnan2016algebraic} is straightforward from the definition.
\begin{fact}
	$d_I \leq d_B$.
\end{fact}

\subsection{Interval decomposable modules}
\begin{sloppypar}
	Persistence modules whose indecomposables are interval modules (Definition \ref{interval-def}) are called {\em interval decomposable modules}, see for example \cite{botnan2016algebraic}. To account for the boundaries of free modules, we enrich the poset $\Rn$ by adding points at $\pm\infty$ and consider the poset $\bar{\Real}^n=\bar{\Real}\times\ldots\times\bar{\Real}$ where $\bar{\Real}=\Real\cup\set{\pm\infty}$ with the usual additional rule $a\pm\infty=\pm\infty$.
\end{sloppypar}

\begin{definition}
	An interval is a subset $\emptyset \neq I \subset \bar{\Real}^n$ that satisfies the following:
	\begin{enumerate}
		\item If $p,q \in I$ and $p \leq r \leq q$, then $r \in I$;
		\item If $p,q \in I$, then there exists a sequence ($p_1, p_2, ... , p_{2m}) \in I$ for some $m \in \mathbb{N}$ such that $p\leq p_1 \geq p_2 \leq p_3 \geq ... \geq p_{2m} \leq q$.
		We call the sequence ($p=p_0, p_1, p_2, ... , p_{2m}, p_{2m+1}=q$) a path from $p$ to $q$ (in $I$). 
	\end{enumerate}
	
	
\end{definition}

Let $\bar{I}$ denote the closure of an interval $I$ in the standard topology of $\RB^n$. 
The lower and upper boundaries of $I$ are defined as
\begin{eqnarray*}
	L(I)&=&\set{x=(x_1,\cdots, x_n)\in \bar{I}\mid \forall y=(y_1,\cdots, y_n) \mbox{ with } y_i< x_i \; \forall i \implies y\notin I}\\
	U(I)&=&\set{x=(x_1,\cdots, x_n)\in \bar{I}\mid \forall y=(y_1,\cdots, y_n) \mbox{ with } y_i> x_i \; \forall i \implies y\notin I}.
\end{eqnarray*}
Let $B(I)=L(I)\cup U(I)$. 
According to this definition, $\RB^n$ is an interval with boundary $B(\RB^n)$ that consists of all the points with at least one coordinate $\infty$. The vertex set $V(\RB^n)$ consists of $2^n$ corner points of the infinitely large 
cube $\RB^n$ with coordinates $(\pm\infty,\cdots, \pm\infty)$.

\begin{definition}[Interval module]
	\label{interval-def}
	An $n$-parameter \textbf{ interval persistence module}, or \textit{interval module} in short, is a persistence module $M$ that satisfies the following condition: for some interval $I_M\subseteq \RB^n$, called the interval of $M$, 
	\begin{equation*}
	M_x =
	\begin{cases}
	
	\mathbb{k} & \mbox{if $x \in I_M$}\\
	0 & otherwise
	\end{cases}
	\qquad 
	\rho^M_{x \rightarrow y} =
	\begin{cases}
	\mathbb{1} & \mbox{if $x,y \in I_M $}\\
	0 & otherwise
	\end{cases}
	\end{equation*}
	
\end{definition}


It is known that an interval module is indecomposable \cite{lesnick2015theory}.
%

\begin{definition}[Interval decomposable module]
	An $n$-parameter {\em interval decomposable module} is a persistence module that can be decomposed into interval modules.
\end{definition}

\begin{definition}[Rectangle]
	A $k$-dimensional rectangle, $0\leq k \leq n$ , or $k$-rectangle, in $\Real^n$,  is a set $I=[a_1,b_1]\times,\cdots, \times[a_n,b_n], a_i,b_i\in \bar{\Real}$, such that, there exists a size $k$ index set $\Lambda \subseteq [n]$, $\forall i\in \Lambda, a_i\neq b_i$, and $\forall j\in [n]-\Lambda, a_j=b_j$. 
\end{definition}
Note that rectangle is an example of interval. A 0-rectangle is a vertex. A 1-rectangle is an edge.



\begin{wrapfigure}{r}{0.55\textwidth}
	\centerline{\includegraphics[height=5.5cm]{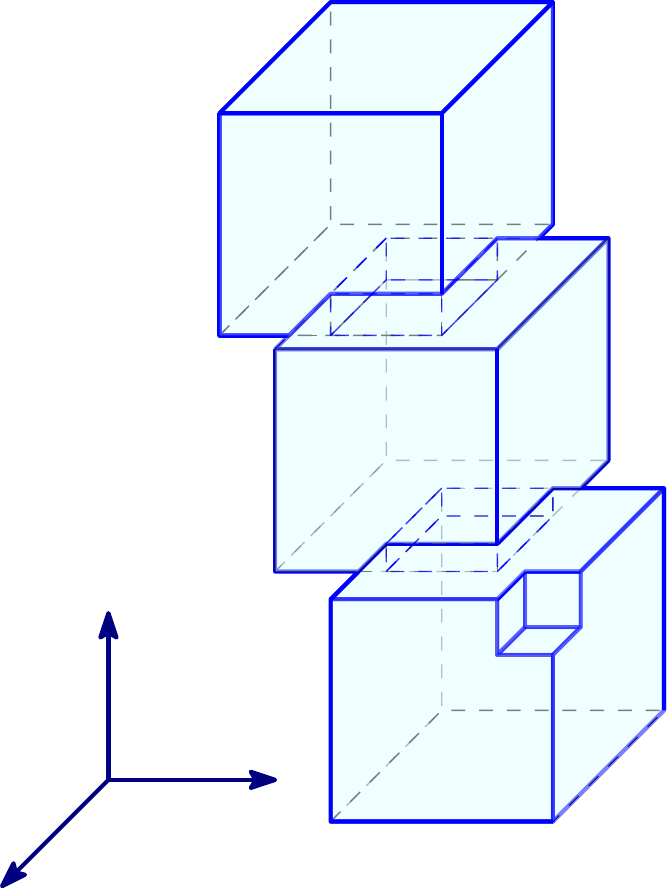}}
\end{wrapfigure}

We say an interval $I\subseteq \Real^n$ is \textbf{discretely presented} if 
it is a finite union of $n$-rectangles.
We also require the boundary of the interval is a $(n-1)$-manifold.
A facet of $I$ is a $(n-1)$-dimensional subset $f=\hat{f}\cap L \subseteq \bar{\Real}^n$ where $\hat{f}=\{x_i=c\}$ is a hyperplane at some standard direction $\vec{e_i}$ in $\Real^n$ and $L$ is either $L(I)$ or $U(I)$. We call such hyperplane $\hat{f}\supseteq f$ the flat of $f$.
We denote the vertex set as $V(I)$ and the facet set as $F(I)$. So the boundary of $I$ is the union of facets. And the vertices of each facet is a subset of $V(I)$.


\begin{wrapfigure}{r}{0.4\textwidth}
	\centerline{\includegraphics[height=4.0cm]{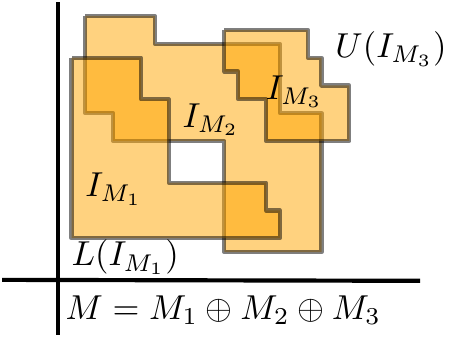}}
\end{wrapfigure}

For $2$-parameter cases, a discretely presented interval $I\subseteq \bar{\Real}^2$ has boundary consisting of a finite set of horizontal and vertical line segments called edges, with end points called vertices, which satisfy the following condition: 
(i) every vertex is incident to either a single horizontal edge or a vertical edge,
(ii) no vertex appears in the interior of an edge.
We denote the set of edges and vertices with $E(I)$ and $V(I)$ respectively. 
We say an $n$-parameter interval decomposable module is finitely presented if it can be decomposed into finitely many interval modules whose intervals are discretely presented (figure on right for an example in 2-parameter cases). They belong to the finitely presented persistence modules as defined in \cite{miller2017data}. In the following, we focus on finitely presented interval decompsable modules.

For an interval module $M$, let $\overline{M}$ be the interval module defined on the closure $\overline{I_M}$. To avoid complication in this exposition, we assume that 
every interval module has closed intervals which is justified by the following proposition (proof in Appendix~\ref{app:miss1}).

\begin{proposition} \label{closure_interleaving}
	$d_I(M, N)=d_I(\overline{M}, \overline{N})$. 
\end{proposition}

\section{Criterion for computing interleaving}
\label{sec:esti_dI}

Given the intervals of the indecomposables (interval modules)  as input,
an approach based on bipartite-graph matching is well known for computing the bottleneck distance $d_B(M,N)$ between two 1-parameter persistence modules $M$ and $N$~\cite{edelsbrunner2010computational}. This approach constructs a bi-partite graph $G$ out of the intervals of $M$ and $N$ and their pairwise interleaving distances including the distances to zero modules. If these distance computations take $O(C)$ time in total, the algorithm for computing $d_B$ takes time $O(m^{\frac{5}{2}}\log m+C)$ if $M$ and $N$ together have $m$ indecomposables altogether. Given indecomposables (say computed  by Meat-Axe~\cite{LS07}), this approach is readily extensible to the $n$-parameter modules if one can compute the interleaving distance between any pair of indecomposables including the zero modules. To this end,
we present an algorithm to compute the interleaving distance between two interval modules $M_i$ and $N_j$ with $t_i$ and $t_j$ vertices respectively on their intervals in $O((t_i+t_j)\log (t_i+t_j))$ time. This gives a total time of $O(m^{\frac{5}{2}}\log m +\sum_{i,j}(t_i+t_j)\log(t_i+t_j))=O(m^{\frac{5}{2}}\log m + t^2\log t)$ where $t$ is the number of vertices over all input intervals. 

Now we focus on computing the interleaving distance between two given intervals. Given two intervals $I_M$ and $I_N$ with $t$ vertices, this algorithm searches a value $\delta$ so that there exists two families of linear maps from $M$ to $N_{\rightarrow \delta}$ and from $N$ to $M_{\rightarrow \delta}$ respectively which satisfy both triangular and square commutativity. This search is done with a binary probing. For a chosen $\delta$ from a candidate set of $O(t)$ values, the algorithm determines the direction of the search by checking two conditions called \textit{trivializability} and \textit{validity} on the intersections of modules $M$ and $N$.

We first present an algorithm for $2$-parameter case since it is more intuitive and the algorithm is relatively simpler. Nevertheless, most of the definitions and claims in this chapter are developed for general $n$-parameter case except a few which are presented for the 2-parameter case first and are generalized later for the n-parameter case. The ones specialized for the 2-parameter case are clearly marked so.

\begin{definition} [Intersection module]
	For two interval modules $M$ and $N$ with intervals $I_M$ and $I_N$ respectively
	let $I_Q = I_M\cap I_N$, which is a disjoint union of intervals, $\coprod I_{Q_i}$.  
	The intersection module $Q$ of $M$ and $N$ is $Q=\bigoplus Q_i$, where $Q_i$ is the interval module with interval $I_{Q_i}$. That is,
	\begin{equation*}
	Q_x= 
	\begin{cases}
	\field{k}  & \mbox{if $x\in I_M\cap I_N$}\\
	0   & otherwise
	\end{cases} \quad \mbox{and for $x\leq y,\,$ }
	\rho^{Q}_{x\rightarrow y}=
	\begin{cases}
	\mathbb{1}  & \mbox{if $x,y \in I_M \cap I_N$}\\
	0   & otherwise
	\end{cases}
	\end{equation*}
\end{definition}
From the definition we can see that the support of $Q$, $supp(Q)$, is $I_M\cap I_N$. 
We call each $Q_i$ an intersection component of $M$ and $N$.  Write $I:=I_{Q_i}$ and consider $\phi:M\rightarrow N$ to be any morphism in the following proposition which says that $\phi$ is constant on $I$.


\begin{proposition}\label{prop_conn}
	$\phi|_I \equiv a\cdot \Id$ for some $a\in\field{k}=\mathbb{Z}/2$.
\end{proposition}

\begin{proof}
	\[ \begin{tikzcd}
	M_{p_{i}} \arrow{r}{\Id}  \arrow[d, dashed, "\phi_{p_i}"']   &M_{p_{i+1}}  \arrow[d, dashed, "\phi_{p_{i+1}}"]
	&M_{p_{i}}   \arrow[d, dashed, "\phi_{p_{i}}"']   &M_{p_{i+1}} \arrow{l}[swap]{\Id} \arrow[d, dashed, "\phi_{p_{i+1}}"]\\%
	N_{p_{i}} \arrow{r}[swap]{\Id}              &N_{p_{i+1}} 
	&N_{p_{i}}              &N_{p_{i+1}} \arrow{l}{\Id}
	\end{tikzcd}
	\]
	For any $x,y \in I$, consider a path $(x=p_0, p_1, p_2, ... , p_{2m}, p_{2m+1}=y)$ in $I$ from $x\ to\ y$ and the commutative diagrams above for $p_i\leq p_{i+1}$ (left) and $p_i\geq p_{i+1}$(right) respectively. Observe that  $\phi_{p_i}=\phi_{p_{i+1}}$ in both cases due to the commutativity. Inducting on $i$, we get that $\phi(x)=\phi(y)$.
\end{proof}

\begin{sloppypar}
	\begin{definition} [Valid intersection]
		An intersection component $Q_i$  is $(M,N)\textit{-valid}$ if for each $x\in I_{Q_i}$ the following two conditions hold (see Figure~\ref{fig:valid_intersection}):
		$$
		\mbox{(i) } y \leq x \mbox{ and } y \in I_M \implies y \in I_N, \mbox{ and (ii) } z \geq x \mbox{ and } z \in I_N \implies z \in I_M
		$$
	\end{definition}
\end{sloppypar}

\begin{proposition} \label{prop:valid_0}
	
	Let $\{Q_i\}$ be a set of intersection components of $M$ and $N$ with intervals $\{I_{Q_i}\}$. Let $\{\phi_x\}:M\rightarrow N$ be the family of linear maps defined as $\phi_x=\Id$ for all $x \in I_{Q_i}$ and $\phi_x=0$ otherwise. Then $\phi$ is a morphism if and only if every $Q_i$ is $(M,N)$\textit{-valid}.
	
	
\end{proposition}
See the proof in Appendix~\ref{app:miss1}.

\begin{figure}
	\centerline{\includegraphics[height=3.5cm]{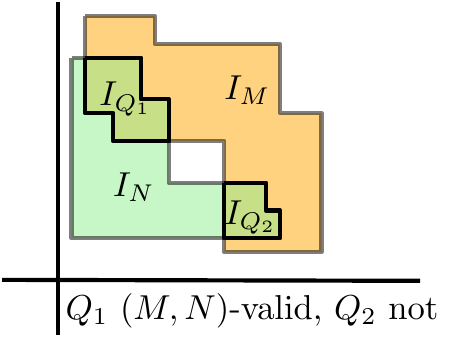}}
	\caption{Examples of a valid intersection and a invalid intersection.}
	\label{fig:valid_intersection}
\end{figure}





From the definition of boundaries of intervals, the following proposition is immediate.

\begin{proposition} 
	Given an interval $I$ and any point $x=(x_1,\cdots, x_n)\in I\setminus (I\cap B(\RB^n))$, we have $x\in L(I) \iff \forall \epsilon>0, x-\vec{\epsilon}\notin I$. Similarly, we have $x\in U(I) \iff \forall \epsilon>0, x+\vec{\epsilon}\notin I$. 
	
\end{proposition}

\begin{definition} [Diagonal projection and distance]
	Let $I$ be an interval and $x\in \RB^n$. Let $\Delta_x=\{x+ \vec{\alpha}\mid  \alpha \in \Real \}$ denote the line called {\em diagonal} with slope $1$ that passes through $x$. We define (see Figure~\ref{fig:dl}) 
	\begin{equation*}
	\dl(x,I)=
	\begin{cases}
	\min_{y\in \Delta_x\cap I}\set{d_\infty(x, y):=|x-y|_{\infty}} \mbox{ if } \Delta_x\cap I \neq \emptyset\\
	+\infty \mbox{ otherwise.}
	\end{cases}
	\end{equation*}
	In case $\Delta_x\cap I\neq \emptyset$, define $\pi_I(x)$, called the projection point of $x$ on $I$, to be the point $y\in{\Delta_x\cap I}$ where  $\dl(x,I)=d_\infty(x, y)$. 
	
	
\end{definition}
Note that $\forall \alpha\in \Real, \quad  \pm \infty + \alpha = \pm \infty$. Therefore, for $x\in V(\RB^n)$, the line collapses to a single point. In that case, $\dl(x, I)\neq+\infty$ if and only if $x\in I$, which means $\pi_I(x)=x$.

\begin{figure}
	\centering
	\includegraphics{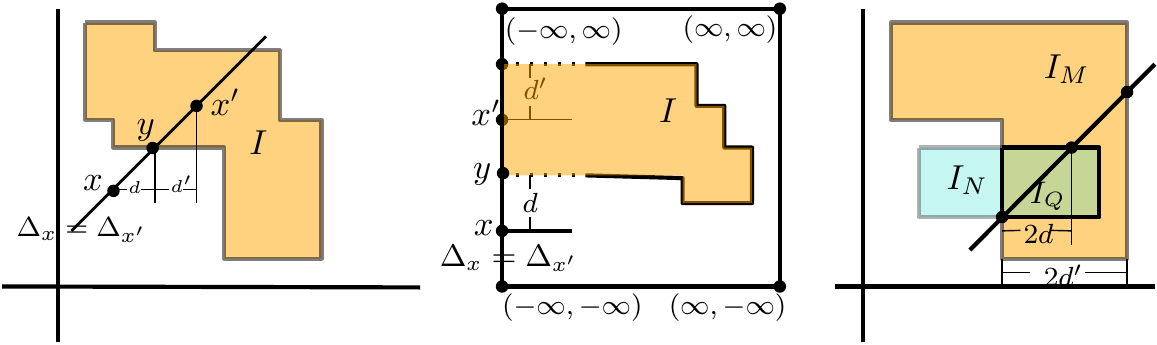}
	\caption{$d=\dl(x,I)$, $y=\pi_I(x)$, $d'=\dl(x',L(I))$ (left); $d=\dl(x,I)$ and $d'=\dl(x',U(I))$ are defined on the left edge of $B(\RB^2)$ (middle); $Q$ is $d'_{(M,N)}$- and $d_{(N,M)}$-trivializable (right) }
	\label{fig:dl}
\end{figure}


Notice that upper and lower boundaries of an interval are also intervals by definition. With this understanding,
following properties of $\dl$ are obvious from the above definition.

\begin{fact}  
	
	\begin{enumerate}\label{fact:dl_distance} 
		\item [(i)] For any $x \in I_M$,
		$$\dl(x, U(I_M))=\sup_{\delta\in \RB}\{x+\vec{\delta} \in I_M\} \mbox{ and } \dl(x, L(I_M))=\sup_{\delta\in \RB}\{x-\vec{\delta} \in I_M\}.$$
		\item[(ii)] Let $L=L(I_M)$ or $U(I_M)$ and let $x,x'$ be two points such that $\pi_L(x), \pi_L(x')$ both exist. If $x$ and $x'$ are on the same facet or the same diagonal line, then $|\dl(x, L)-\dl(x', L)|\leq d_\infty(x, x')$. 
	\end{enumerate}

\end{fact}




Set $VL(I):=V(I)\cap L(I)$, $EL(I):=E(I)\cap L(I)$, $VU(I):=V(I)\cap U(I)$, and $EU(I):=E(I)\cap U(I)$. Following proposition is proved in Appendix~\ref{app:miss1}.

\begin{proposition} \label{prop_valid_intersect}
	For an intersection component $Q$ of $M$ and $ N$ with interval $I$, the following conditions are equivalent:
	\begin{enumerate}[label={(\arabic*)}]
		\item $Q$ is $(M, N)$\textit{-valid}.
		\item $L(I) \subseteq L(I_M)$ and $U(I) \subseteq U(I_N)$.
		\item $VL(I) \subseteq L(I_M)$ and $VU(I) \subseteq U(I_N)$.    
	\end{enumerate}
	
\end{proposition}







\begin{definition} [Trivializable intersection]
	Let $Q$ be a connected component of the intersection of two modules $M$ and $N$. 
	For each point $x\in I_Q$, define
	$$d^{(M,N)}_{triv}(x) =\max\{\dl(x, U(I_M))/2, \dl(x, L(I_N))/2)\}.$$ For $\delta \geq 0$, we say a point $x$ is $\delta_{(M, N)}$\textit{-trivializable} if $d^{(M,N)}_{triv}(x) < \delta$.
	We say an intersection component $Q$ is $\delta_{(M,N)}$\textit{-trivializable} if each point in $I_Q$ is $\delta_{(M,N)}$\textit{-trivializable} (Figure~\ref{fig:dl}). We also denote $d_{triv}^{(M, N)}(I_Q):=\sup_{x\in I_Q}{\{d_{triv}^{(M, N)}(x)\}}$
\end{definition}




Following proposition discretizes the search for trivializability (proof in Appendix \ref{app:miss1}).
\begin{proposition} \label{prop:triv}
	
	An intersection component $Q$ is $\delta_{(M,N)}$\textit{-trivializable} if and only if every vertex of $Q$ is $\delta_{(M,N)}$\textit{-trivializable}.
\end{proposition}



Recall that for two modules to be $\delta$-interleaved, we need two families of linear maps satisfying both triangular commutativity and square commutativity. For a given $\delta$, Theorem~\ref{di_criteria} below provides criteria which ensure that such linear maps exist. In our algorithm, we make sure that these criteria are verified.


Given an interval module $M$ and the diagonal line $\Delta_x$ for any $x\in \RB^n$, there is a $1$-parameter persistence module $M|_{\Delta_x}$ which is the functor restricted on the poset $\mathbf{\Delta}_x$ as a subcategory of $\bar{\mathbb{R}^n}$. We call it a \textit{1-parameter slice} of $M$ along $\Delta_x$. 
Define 
$$\delta^* = \inf_{\delta\in\RB} \set{\delta: \forall x\in \RB^n, M|_{\Delta_x} \mbox{ and } N|_{\Delta_x}\mbox{ are } \delta\mbox{-interleaved}}$$ 

Equivalently, we have

$$
\delta^*=\sup_{x\in \RB^n}\set{d_I(M|_{\Delta_x}, N|_{\Delta_x})}
$$

We have the following Proposition and Corollary from the equivalent definition of $\delta^*$.


\begin{proposition} \label{prop_tria_comm} 
	For two interval modules $M, N$ and $\delta > \delta^*  \in \Real^+$,
	there exist two families of linear maps $\phi=\set{\phi_x: M_x \rightarrow N_{(x+\delta)}}$ and $\psi=\set{\psi_x: N_x \rightarrow M_{(x+\delta)}}$ such that for each $x\in\RB^n$, the 1-parameter slices $M|_{\Delta_x}$ and $N|_{\Delta_x}$ are $\delta$-interleaved by the linear maps $\phi|_{\Delta_x}$ and $\psi|_{\Delta_x}$.
	
\end{proposition}

\begin{corollary} \label{cor:di_delta_star_}
	$d_I(M,N)\geq \delta^*$
\end{corollary}





\begin{theorem} \label{di_criteria}

	For two interval modules $M$ and $N$, $d_I(M, N)\leq \delta$ if and only if both of the following two conditions are satisfied:
	
	(i) $\delta \geq \delta^*$,
	
	\begin{sloppypar}
		(ii) $\forall\delta' > \delta$, each intersection component of $M$ and $N_{\rightarrow \delta'}$ is either $(M,N_{\rightarrow \delta'})$\textit{-valid} or $\delta_{(M,N_{\rightarrow \delta'})}$\textit{-trivializable}, and each intersection component of $ M_{\rightarrow \delta'}$ and $N$ is either $(N,M_{\rightarrow \delta'})$\textit{-valid} or $\delta_{(N,M_{\rightarrow \delta'})}$\textit{-trivializable}
	\end{sloppypar}
	
\end{theorem}

\begin{proof}
	Recall that, by definition, $d_I(M,N)\leq \delta$ if and only if $\forall \delta' > \delta, M, N$ is $\delta'$-interleaved.\\
	
	$\implies$ direction: Given $M$ and $N$ are $\delta$-interleaved. Condition (i) follows from Corollary~\ref{cor:di_delta_star_} directly. Consider condition (ii).
	By definition of interleaving, $\forall \delta' > \delta$, we have two families of linear maps $\{\phi_x\}$ and $\{\psi_x\}$ which satisfy both triangular and square commutativities. Let the morphisms between the two persistence modules constituted by these two families of linear maps be $\phi=\{\phi_x\}$ and $\psi=\{\psi_x\}$ respectively.
	For each intersection component $Q$ of $M$ and $N_{\rightarrow \delta'}$ with interval $I:=I_Q$, consider the restriction $\phi|_I$. By Proposition \ref{prop_conn}, $\phi|_I$ is constant, that is, $\phi|_I\equiv 0 $ or
	$\mathbb{1}$. 
	If $\phi|_I\equiv \mathbb{1}$, by Proposition \ref{prop:valid_0}, $Q$ is $(M,N_{\rightarrow \delta'})$\textit{-valid}.
	If $\phi|_I\equiv 0$, by the triangular commutativity of $\phi$, we have that $\rho^M_{x\rightarrow x+2\vec\delta'}=\psi_{x+\vec\delta'}\circ\phi_x=0$ for each point $x\in I$. That means $x+2\vec{\delta'} \notin I_M$. By Fact \ref{fact:dl_distance}(i), $\dl(x, U(I_M))/2 < \delta'$. 
	Similarly, $\rho^N_{x-\vec{\delta'}\rightarrow x+\vec{\delta'}}=\phi_x\circ\psi_{x-\vec{\delta'}}=0 \implies x-\vec{\delta'} \notin I_N$, which is the same as to say $x-2\vec{\delta'}\notin I_{N_{\rightarrow \delta'}}$. 
	By Fact \ref{fact:dl_distance}(i), $\dl(x, L(I_{N_{\rightarrow \delta'}}))/2 < \delta'$. 
	So $\forall x\in I$, we have $d^{(M, N_{\rightarrow \delta'})}_{triv}(x) < \delta'$. This means $Q$ is $\delta'_{(M, N_{\rightarrow\delta'})}$\textit{-trivializable}. Similar statement holds for intersection components of $M_{\rightarrow \delta'}$ and $N$.

	$\Longleftarrow$ direction:
	We construct two families of linear maps $\{\phi_x\}, \{\psi_x\}$ as follows: On the interval $I:=I_{Q_i}$  of each intersection component $Q_i$ of $M$ and $N_{\rightarrow \delta'}$, set $\phi|_I\equiv \Id$ if $Q_i$ is $(M, N_{\rightarrow \delta'})$\textit{-valid} and $\phi|_I\equiv 0$ otherwise. Set $\phi_x\equiv 0$ for all $x$ not in the interval of any intersection component. Similarly, construct $\{\psi_x\}$. 
	Note that, by Proposition \ref{prop:valid_0}, $\phi:=\{\phi_x\}$ is a morphism between $M$ and $N_{\rightarrow \delta'}$, and $\psi:=\{\psi_x\}$ is a morphism between $N$ and $M_{\rightarrow \delta'}$. Hence, they satisfy the square commutativity. 
	We show that they also satisfy the triangular commutativity.
	
	We claim that $\forall x\in I_M$, $\rho^M_{x\rightarrow x+2\vec{\delta'}}=\Id\implies x+\vec{\delta'}\in I_N$ and similar statement holds for $I_N$. 
	From condition that $\delta' > \delta \geq \delta^*$ and by proposition \ref{prop_tria_comm}, we know that there exist two families of linear maps satisfying triangular commutativity everywhere, especially on the pair of $1$-parameter persistence modules $M|_{\Delta_x}$ and $N|_{\Delta_x}$. From triangular commutativity, we know that for $\forall x\in I_M$ with $\rho^M_{x\rightarrow x+2\vec{\delta'}}=\Id$, $x+\vec{\delta'} \in I_N$ since otherwise one cannot construct a $\delta$-interleaving between $M|_{\Delta_x}$ and $N|_{\Delta_x}$. So we get our claim. 
	
	Now for each $x\in I_M$ with $\rho^M_{x\rightarrow x+2\vec{\delta'}}=\Id$,  we have $\dl(x, U(I_M))/2 \geq \delta'$ by Fact \ref{fact:dl_distance}, and $x+\vec{\delta'} \in I_{N}$ by our claim. This implies that $x\in I_M\cap I_{N\rightarrow \delta'}$ is a point in an interval of an intersection component $Q_x$ of $M, N_{\rightarrow \delta'}$ which is not $\delta'_{(M,N_{\rightarrow \delta'})}$\textit{-trivializable}. Hence, it is $(M, N_{\rightarrow \delta'})$\textit{-valid} by the assumption. So, by our construction of $\phi$ on valid intersection components, $\phi_x = \Id$. 
	Symmetrically, we have that  $x+\vec{\delta'}\in I_N\cap I_{M\rightarrow \delta'}$ is a point in an interval of an intersection component of $N$ and $M_{\rightarrow \delta'}$ which is not $\delta'_{(N, M_{\rightarrow\delta'})}$\textit{-trvializable} since $\dl(x+\vec{\delta'}, L(I_M))/2 \geq \delta'$. So by our construction of $\psi$ on valid intersection components, $\psi_{x+\vec{\delta'}}=\Id$.
	Then, we have $\rho^M_{x \rightarrow x+2\vec{\delta'}}=\psi_{x+\vec{\delta'}}\circ\phi_x$ for every nonzero linear map $\rho^M_{x \rightarrow x+2\vec{\delta'}}$. The statement also holds for any nonzero linear map $\rho^N_{x \rightarrow x+2\vec{\delta'}}$. Therefore, the triangular commutativity holds.
\end{proof}

Note that the above proof provides a construction of the interleaving maps for any specific $\delta'$ if it exists. 
Furthermore, the interleaving distance $d_I(M,N)$ is the infimum of all $\delta'$ satisfying the two conditions in the theorem, which means $d_I(M, N)$ is the infimum of all $\delta'\geq \delta^*$ satisfying condition 2 in Theorem~\ref{di_criteria}.

\section{Algorithm to compute $d_I$}

In practice, we cannot verify all those infinitely many values $\delta'>\delta^*$ required by Theorem~\ref{di_criteria}. We propose a finite candidate set of potentially possible interleaving distance values and prove later that our final target, the interleaving distance, is always contained in this finite set. Surprisingly, the size of the candidate set is only $O(t)$ with respect to the $t$ number of vertices for 2-parameter interval modules and $O(t^2)$ in higher dimensional case.
We first discuss the 2-parameter case.
\subsection{2-parameter module}

Based on our results, we propose a search algorithm for computing the interleaving distance $d_I(M,N)$ for interval modules $M$ and $N$.

\begin{definition} [Candidate set for 2-parameter cases]
	For two interval modules $M$ and $N$, and for each point $x$ in $I_M \cup I_N$, let
	\begin{eqnarray*}
		D(x)& =& \{\dl(x, L(I_M)), \dl(x, L(I_N)), \dl(x, U(I_M)), \dl(x, U(I_N))\} \mbox{ and}\\
		S& = &\{d \mid d\in D(x) \mbox{ or } 2d\in D(x) \mbox{ for some vertex } x\in V(I_M)\cup V(I_N) \} \mbox { and } \\
		S_{\geq \delta}&:=&\set{d \mid d\geq \delta, d\in S}. 
	\end{eqnarray*}
\end{definition}


\noindent\textbf{Algorithm} {\sc Interleaving} (output: $d_I(M, N)$, input: $I_M$ and $I_N$ with $t$ vertices in total)
\begin{enumerate}
	
	\item Compute the candidate set $S$ and 
	let $\epsilon$ be the half of the smallest difference between any two numbers in $S$.  /* $O(t)$ time */
	\item Compute $\delta^*$; Let $\delta = \delta^*$.  /* $O(t)$ time */
	\item Output $\delta$ after a binary search in $S_{\geq \delta^*}$ by following steps /* $O(\log t$) probes */ 
	\begin{itemize}
		\item let $\delta'=\delta+\epsilon$
		\item Compute intersections $I_M \cap I_{N_{\rightarrow\delta'}}$ and $I_N \cap I_{M_{\rightarrow\delta'}}$. /* $O(t)$ time */
		\item For each intersection component, check if it is valid or trivializable according to Theorem~\ref{di_criteria}. /* $O(t)$ time */
	\end{itemize}
\end{enumerate}

In the above algorithm, the following generic task of computing {\em diagonal span} is performed for several steps. Let $L$ and $U$ be any two chains of vertical and horizontal edges that are both $x$- and $y$-monotone. Assume that $L$ and $U$ have at most $t$ vertices. Then, for a set $X$ of $O(t)$ points in $L$, one can compute the intersection of $\Delta_x$ with $U$ for every $x\in X$ in $O(t)$ total time. The idea is to first compute by a binary search a point $x$ in $X$ so that $\Delta_x$ intersects $U$ if at all. Then, for other points in $X$, traverse from $x$ in both directions while searching for the intersections of the diagonal line with $U$ in lock steps.

Now we analyze the complexity of the algorithm {\sc Interleaving}. The candidate set, by definition, has $O(t)$
values which can be computed in $O(t)$ time by the diagonal span procedure. Proposition~\ref{prop:delstar-search} shows that $\delta^*$ is in $S$ and can be determined by computing the one dimensional interleaving distances $d_I(M|_{\Delta_x},N|_{\Delta_x})$ for diagonal lines passing through $O(t)$ vertices of $I_M$ and $I_N$. This can be done in $O(t)$ time by diagonal span procedure. Once we determine $\delta^*$, we search for $\delta=d_I(M,N)$ in the truncated set $S_{\delta\geq \delta^*}$ to satisfy the first condition of Theorem~\ref{di_criteria}. Intersections between two polygons $I_M$ and $I_N$ bounded by $x$- and $y$-monotone chains can be computed in $O(t)$ time by a simple traversal of the boundaries. The validity and trivializability of each intersection component can be determined in time linear in the number of its vertices due to Proposition~\ref{prop_valid_intersect} and Proposition~\ref{prop:triv} respectively. Since the total number of intersection points is $O(t)$, validity check takes $O(t)$ time in total. The check for trivializabilty also takes $O(t)$ time if one uses the diagonal span procedure. 
Taking into account $O(\log t)$ probes, the total time complexity of the algorithm becomes $O(t\log t)$.

Proposition~\ref{prop:delstar-search} below says that $\delta^*$ is determined by a vertex in $I_M$ or $I_N$ and $\delta^*\in S$. It follows from applying Proposition~\ref{prop:deltastar-search_n} to the case $n=2$.

\begin{proposition}[2-parameter case] 
	
	(i) $\delta^*=\max_{x\in V(I_M)\cup V(I_N)}\set{d_I(M|_{\Delta_x}, N|_{\Delta_x})}$, (ii) $\delta^*\in S$.
	\label{prop:delstar-search}
\end{proposition}

The correctness of the algorithm {\sc Interleaving} already follows from Theorem~\ref{di_criteria} as long as the candidate set contains the distance $d_I(M,N)$. The following concept of stable intersections helps us to establish this result.

\begin{definition} [Stable intersection]
	Let $Q$ be an intersection component of $M$ and $N$. We say $Q$ is stable if $M$ and $N$ do not intersect at $Q$ transversally. 
	This means that any point $x\in B(I_Q)$ cannot be in the intersection of any two parallel facets of $I_M$ and $I_N$.
\end{definition}

From Proposition \ref{touching_pt_intersection} and Corollary \ref{cor:valid_intersection} in Appendix~\ref{app:miss1}, we have the following proposition.
\begin{proposition} \label{prop:stable_intersection_iff_dnotinS}
	$d\notin S$ if and only if each intersection component of $M,N_{\rightarrow d}$, and $N_{\rightarrow d}, M$ is stable.
\end{proposition}

The main property of a stable intersection component $Q$ of $M$ and $N$ is that if we shift one of the interval module, say $N$, to $N_{\rightarrow \epsilon}$ continuously for some small value $\epsilon\in \Real^+$, the interval $I_{Q^\epsilon}$ of the intersection component $Q^\epsilon$ of $M$ and $N_{\rightarrow \epsilon}$ changes continuously. Next proposition follows directly from the stability of intersection components. 

\begin{proposition} \label{prop:stable_intersection_property0}
	For a stable intersection component $Q$ of $M$ and $N$, there exists a positive real $\delta \in \Real^+$ so that the following holds:
	
	For each $\epsilon \in (-\delta, +\delta)$, there exists a unique intersection component $Q^\epsilon$ of $M$ and $N_{\rightarrow \epsilon}$ so that it is still stable and $I_{Q^\epsilon}\cap I_Q\neq \emptyset$.
	Furthermore, there is a bijection $\mu_\epsilon:V(I_Q)\rightarrow V(I_{Q^\epsilon})$
	so that $\forall x\in V(I_Q)$, $x$ and $\mu_\epsilon(x)$ are on the same facet and $d_\infty (\mu_\epsilon(x), x)=\epsilon$. 
	We call the set $\set{Q^\epsilon\mid \epsilon\in (-\delta, +\delta)}$ a stable neighborhood of $Q$.
	
\end{proposition}

\begin{corollary} \label{cor:stable_intersection_property}
	For a stable intersection component $Q$, we have:
	
	(i) $Q$ is $(M,N)$-valid iff each $Q^\epsilon$ in the stable neighborhood is $(M,N_{\rightarrow \epsilon})$-valid. 
	
	(ii) If $Q$ is $d_{(M,N)}$-trivializable, then $Q^\epsilon$ is $(d+2\epsilon)_{(M,N_{\rightarrow \epsilon})}$-trivializable.
\end{corollary}

\begin{proof}
	(i): Let $Q^\epsilon$ be any intersection component in a stable neighborhood of $Q$. 
	We know that if $Q$ is ($M,N$)-valid, then $VL(I_Q)\subseteq L(I_M)$ and $VU(I_Q)\subseteq U(I_N)$. By Proposition \ref{prop:stable_intersection_property0}, $\mu_\epsilon(VL(I_Q))=VL(I_{Q^\epsilon})\subseteq L(I_M)$ and $\mu_\epsilon(UL(I_Q))=UL(I_{Q^\epsilon})\subseteq L(I_{N\rightarrow \epsilon})$. So  $Q^\epsilon$ is $(M,N_{\rightarrow \epsilon})$-valid. Other direction of the implication can be proved by switching the roles of $Q$ and $Q^\epsilon$ in the above argument.
	
	(ii): From Proposition \ref{prop:stable_intersection_property0}, we have that $\forall x' \in V(I_{Q^\epsilon})$, there exists a point $x\in V(I_Q)$ so that $x$ and $x'$ are on some horizontal, vertical, or diagonal line ($\Delta_x$), and $d_\infty(x, x')\leq \epsilon$. Then, by Fact \ref{fact:dl_distance}(ii), one observes $$d^{(M, N_{\rightarrow \epsilon})}_{triv}(x) \leq d^{(M, N_{\rightarrow \epsilon})}_{triv}(x')+\epsilon \leq d^{(M, N)}_{triv}(x)+2\epsilon < d+2\epsilon.$$ Therefore, $Q^\epsilon$ is $(d+2\epsilon)_{(M, N_{\rightarrow \epsilon})}$-trivializable. 
\end{proof}

\begin{proposition}\label{prop:dtriv_in_S}
	For any intersection component $Q$ of $M$ and $N$, $d_{triv}^{M, N}(I_Q)\in S$ for n=2 (2-parameter case).
	
\end{proposition}
\begin{proof}
	By definition of $d_{triv}^{M,N}$, it is not hard to see that $d_{triv}^{M,N}(I_Q)$ is realized by some $x\in B(I_Q)$. Furthermore, by Proposition~\ref{prop:distance_bound_2}, it can be realized by some $x\in V(I_Q)$. Let $f \in U(I_M)\cup L(I_N)$ be the facet such that $d_{triv}^{M,N}(x)=\dl(x, f)=\dl(x,x')$ where $x'=\pi_f(x)$. That is $d_{triv}^{M,N}(x)$ is realized by the distance between $x$ and $f$. Then by the definition of interval, one can observe that $x$ must be contained in $n$ non-parallel facets from either $F(I_M)$ or $F(I_N)$. So there is at least one facets containg $x$ which is parallel with $f$. By Corollary~\ref{cor:distance_bound}, we get the conclusion.
	
\end{proof}

By definition of $d_{triv}^{M,N}$ and the candidate set $S$, we get the following corollary.

\begin{corollary}\label{cor:dtriv_in_S}
	For any intersection component $Q$ of $M$ and $N_{\rightarrow d}$, $d \in S$ where $d=d_{triv}^{M, N_{\rightarrow d}}(I_Q)$.
\end{corollary}

Note that here the set $S$ is defined for the original modules $M$ and $N$ without any shifting.

\begin{theorem} \label{thm:dIinS}
	$d_I(M, N)\in S$. 
\end{theorem}

\begin{proof}
	Suppose that $d=d_I(M,N)\not\in S$. Let $d^*$ be the largest value in S satisfying $d^* \leq d$. Note that $d\in S$ if and only if $d=d^*$. Then, $d^* < d$ by our assumption that $d \notin S$. 
	
	By definition of interleaving distance, we have $\forall d' > d$, there is a $d'$-interleaving between $M$ and $N$, and $\forall d'' < d$, there is no $d''$-interleaving between $M$ and $N$. By Proposition \ref{prop:delstar-search}(ii), one can see that $ \delta^* \leq d^* < d$. So, to get a contradiction, we just need to show that there exists $d''$, $d^* < d''< d$, satisfying the condition 2 in Theorem \ref{di_criteria}.
	
	Let $Q$ be any intersection component of $M, N_{\rightarrow d}$ or $N, M_{\rightarrow d}$.
	Without loss of generality, assume $Q$ is an intersection component of $M$ and $N_{\rightarrow d}$. By Proposition \ref{prop:stable_intersection_iff_dnotinS}, $Q$ is stable. We claim that there exists some $\epsilon > 0$ such that $Q^{-\epsilon}$ is an intersection component of $M$ and $N_{\rightarrow d-\epsilon}$ in a stable neighborhood of $Q$, and $Q^{-\epsilon}$ is either $(M,N_{\rightarrow d-\epsilon})$-valid or $(d-\epsilon)_{(M,N_{\rightarrow d-\epsilon})}$-trivializable.
	
	Let $\epsilon>0$ be small enough so that $Q^{+\epsilon}$ is a stable intersection component of $M$ and $N_{\rightarrow d+\epsilon}$ in a stable neighborhood of $Q$. 
	By Theorem \ref{di_criteria}, $Q^{+\epsilon}$ is either $(M, N_{\rightarrow (d+\epsilon)})$-valid or $(d+\epsilon)_{(M, N_{\rightarrow (d+\epsilon)})}$-trivializable. 
	If $Q^{+\epsilon}$ is $(M, N_{\rightarrow (d+\epsilon)})$-valid, then by Corollary \ref{cor:stable_intersection_property}(i), any intersection component in a stable neighborhood of $Q$ is valid, which means there exists $Q^{-\epsilon}$ that is $(M,N_{\rightarrow d-\epsilon})$-valid for some $\epsilon > 0$.
	Now assume $Q^{+\epsilon}$ is not $(M, N_{\rightarrow (d+\epsilon)})$-valid. Then, $\forall \epsilon > 0$, $Q^{+\epsilon}$ is ${(M, N_{\rightarrow (d+\epsilon)})}$-trivializable, By Proposition \ref{prop:triv} and \ref{cor:stable_intersection_property}(ii), we have $\forall x\in V(I_Q)$, $d^{(M, N_{\rightarrow d+\epsilon})}_{triv}(x) < d+3\epsilon$, $\forall \epsilon > 0$. 
	Taking $\epsilon \to 0$, we get $\forall x\in V(I_Q)$, $d_{triv}^{(M,N_{\rightarrow d})}(x)\leq d$. 
	We claim that, actually, $\forall x\in V(I_Q)$, $d_{triv}^{(M,N_{\rightarrow d})}(x)< d$. If the claim were not true, some point $x\in V(I_Q)$ would exist so that $d_{triv}^{(M,N_{\rightarrow d})}(x) = d$. 
	By Corollary~\ref{cor:dtriv_in_S}, we have $d\in S$, contradicting $d\neq d^*$.
	
	
	Now by our claim and Proposition \ref{prop:triv}, $Q$ is $d_{(M,N_{\rightarrow d})}$-trivializable where 
	$d>d^*\geq \max_{x\in V(I_Q)}\set{d^{(M, N_{\rightarrow d})}_{triv}(x)}$. Let $\delta=d-d^*$ and $\epsilon=\delta/4$. Since $ d-\epsilon = d-\delta/4 > d-\delta/2=d-\delta+2\cdot\delta/4= d^*+2\epsilon$ and $d^* \geq \max_{x\in V(I_Q)}\set{d^{(M, N_{\rightarrow d})}_{triv}(x)}$, we have  $d>d^*$ and $d-\epsilon > \max_{x\in V(I_Q)}\set{d^{(M, N_{\rightarrow d})}_{triv}(x)}+2\epsilon$. Therefore, by Corollary \ref{prop:triv}, $Q^{-\epsilon}$ is $(d-\epsilon)_{(M, N_{\rightarrow d-\epsilon})}$-trivializable.

	The above argument shows that 
	there exists a $d''$-interleaving where $d''=d-\epsilon<d$, reaching a contradiction.
\end{proof}

\begin{remark}
	Our main theorem and the algorithm based on it consider the persistence modules defined over $\Real^n$ ($n=2$ in this subsection). In practice, we often deal with persistence modules defined on a discrete grid like $\Int^n$. In this case, we can consider the embedded persistence modules defined over $\Int^n$ into $\Real^n$ and apply our theorem and algorihtm accordingly.
\end{remark}

\subsection{$n$-parameter module}
To extend our results to n-parameter case, we need the following definitions and propositions. Most of them are the extensions of the original ones in 2-parameter case. Also, the algorithm needs adjustments.

To make sure $d_I\in S$, we need to change the set $S$ to be slightly larger but still with finite size.
\begin{definition}[Extended candidate set for $(n>2)$-parameter case]
	For two interval modules $M$ and $N$, and for each point $x$ in $I_M \cup I_N$, let
	\begin{eqnarray*}
		\hat{D}(x)&=&\{\dl(x, \hat{f})\mid f\in F(I_M)\cup F(I_N) \} \mbox{, recall that $\hat{f}$ is the flat of $f$ in $\Real^n$.}\\
		\hat{S}& = &\{d \mid d\in \hat{D}(x) \mbox{ or } 2d\in \hat{D}(x) \mbox{ for some vertex } x\in V(I_M)\cup V(I_N) \} \mbox { and } \\
		\hat{S}_{\geq \delta}&:=&\set{d \mid d\geq \delta, d\in \hat{S}}. 
	\end{eqnarray*}
\end{definition}

For any facet $f\in F(I_M)$, let $/f/:=\{x+\vec{t}\mid x\in f, t\in \Real\}$. This can be viewed as a translate of $f$ along the diagonal line direction. Define a set $\bar{V}:=\{V(/f/\cap g)\mid f,g\in F(I_M)\cup F(I_N) \}$.
Observe that, since each facet $g$ belongs to a hyperplane $\hat{g}=\{x_i=c\}\subseteq\Real^n$ for some $i$ and a constant $c\in \bar{\Real}$, the intersection $/f/\cap g$ is a convex set in $g$ with boundary edges $E(/f/\cap g)$ consisting of edges only along a standard direction $\vec{e_i}$ or the direction of the projection of $\vec{e}=(1,\cdots, 1)$ onto $\hat{g}$.

We use the following important fact.

\begin{fact}\label{fact:vertex_bound}
	$\forall x\in (/f/\cap g) - V( /f/\cap g)$, $\exists y,z\in V(/f/\cap g)$,  $y<x<z$.
\end{fact}

We also have the following proposition.

\begin{proposition} [Extension of Proposition~\ref{prop:distance_bound_2} for $(n>2)$-parameter case]
	Let $M$ and $N$ be two interval modules. 
	Given any point $x\in B(I_M)$ and any $L\in\set{L(I_M),U(I_M),L(I_N),U(I_N)}$, 
	with $x'=\pi_L(x)$ existing, let $d_x=\dl(x, L)$, $F_x$ and $F_{x'}$ be the two facets containing $x$ and $x'$ respectively. Then there exist (not necessarily distinct) $y, z\in V(F_x)\cup V(F_{x'})$ and $d_y\in \hat{D}(y), d_z\in \hat{D}(z)$ such that $d_y\leq d_x\leq d_z$. 
	\label{prop:distance_bound}
\end{proposition}

\begin{proof}
	Note that the facet ${F_{x'}}$ belongs to the hyperplane $\hat{F_{x'}}=\{x_i=c\}\subset \bar{\Real}^n$ for some $i\in \mathbb{N}$ and constant $c\in \bar{\Real}$. Consider the $\bar{\Real}$-valued function $\phi: F_x\rightarrow \bar{\Real}$ given by $\phi(w)=\dl(w, \hat{F_{x'}})$. 
	Observe that $\phi(w)=|c-w_i|$. So this function is a linear function on $F_{x'}$.
	By the property of linearity, we have that the maximum and minimum are achieved in $V(F_x)$.

\end{proof}



The following three statements all depend on the extension of the candidate set $\hat{D}$ and $\hat{S}$, and the extended Proposition~\ref{prop:distance_bound}. The proofs are almost the same except that, in order to apply the extended version of propositions in n-parameter cases, we have to replace  $\hat{D}$ and $\hat{S}$ with $D$ and $S$.

\begin{proposition} [Extension of Proposition ~\ref{prop:dtriv_in_S}]
	For any intersection component $Q$ of $M$ and $N$, $d_{triv}^{M, N}(I_Q)\in \hat{S}$
	\label{prop:dtriv_in_S_n}
\end{proposition}

\begin{corollary}[Extension of Corollary~\ref{cor:dtriv_in_S}]
	For any intersection component $Q$ of $M$ and $N_{\rightarrow d}$,  we have $d \in \hat{S}$ where $d=d_{triv}^{M, N_{\rightarrow d}}(I_Q)$.
	\label{cor:dtriv_in_S_n}
\end{corollary}

\begin{theorem} [Extension of Theorem~\ref{thm:dIinS}]
	$d_I(M, N)\in \hat{S}$. 
	\label{thm:dIinS_n}
\end{theorem}

\noindent
\begin{proposition}[Extension of Proposition~\ref{prop:delstar-search} for $(n>2)$-parameter case]
	$(i)~ \delta^*=\max_{x\in \bar{V}}\set{d_I(M|_{\Delta_x}, N|_{\Delta_x})},~(ii)~ \delta^*\in S$.
	\label{prop:deltastar-search_n}
\end{proposition}

\begin{proof} First, we show (i).
	By definition of $\delta^*$, the claim is equivalent to showing that
	$$
	\delta \geq \delta^* \iff \delta \geq \max_{x\in \bar{V}}\set{d_I(M|_{\Delta_x}, N|_{\Delta_x})}. 
	$$
	
	We observe the following chain of equivalences.
	\begin{eqnarray*}
		\delta \geq \delta^*
		&\iff& 
		\mbox{ for every pair $M|_{\Delta_x}, N|_{\Delta_x}$,  
			$\delta \geq d_I(M|_{\Delta_x}, N|_{\Delta_x})$}\\
		&\iff&
		\begin{cases}
			\forall x \in B(I_M), \\
			\dl(x, U(I_M))> 2\delta \implies x + \vec{\delta} \in I_N 
			\mbox{ and }\pi_{U(I_M)}(x)-\vec{\delta}\in I_N,
			\\
			\dl(x, L(I_M))> 2\delta \implies x - \vec{\delta} \in I_N 
			\mbox{ and }\pi_{L(I_M)}(x)+\vec{\delta}\in I_N.
			\\
			\forall y \in B(I_N), \\
			\dl(y, U(I_N))>2\delta \implies y + \vec{\delta} \in I_M 
			\mbox{ and }\pi_{U(I_N)}(y)-\vec{\delta}\in I_M,
			\\
			\dl(y, L(I_N))>2\delta \implies y - \vec{\delta} \in I_M
			\mbox{ and }\pi_{U(I_N)}(y)+\vec{\delta}\in I_M.
			
		\end{cases}\\
		&\iff&
		\begin{cases}
			\forall x \in \bar{V}\cap V(I_M),\\
			\dl(x, U(I_M))> 2\delta \implies x + \vec{\delta} \in I_N     \mbox{ and }\pi_{U(I_M)}(x)-\vec{\delta}\in I_N,\\
			
			\dl(x, L(I_M))> 2\delta \implies x - \vec{\delta} \in I_N      \mbox{ and }\pi_{L(I_M)}(x)+\vec{\delta}\in I_N.
			\\
			\forall y \in \bar{V}\cap V(I_N),\\			
			\dl(y, U(I_N))>2\delta \implies y + \vec{\delta} \in I_M 
			\mbox{ and }\pi_{U(I_N)}(y)-\vec{\delta}\in I_M,
			\\
			
			\dl(y, L(I_N))>2\delta \implies y - \vec{\delta} \in I_M
			\mbox{ and }\pi_{U(I_N)}(y)+\vec{\delta}\in I_M.
		\end{cases}\\
		&\iff&
		\delta \geq \max_{x\in V(I_M)\cup V(I_N)}\set{d_I(M|_{\Delta_x}, N|_{\Delta_x})}.
	\end{eqnarray*}

	The first two and the last equivalences follow from the definition of interleaving distance and Proposition~\ref{prop:1d_interleaving_distance}.
	The $\implies$ direction of the third equivalence follows trivially from the fact that $\bar{V}\cap V(I_M) \subseteq B(I_M)$ and $\bar{V}\cap V(I_N) \subseteq B(I_N)$. For the $\Longleftarrow$ direction, 
	we show that if the implications $x\pm2\vec{\delta} \in I_M \implies x\pm\vec{\delta} \in I_N$ hold for every point  $x\in \bar{V}\cap V(I_M)$,
	then they also hold for every point in $B(I_M)$. 
	Similarly, one can show if the implications $x\pm2\vec{\delta} \in I_N \implies x\pm\vec{\delta} \in I_M$ hold for every point  $x\in \bar{V}\cap V(I_N)$, then they also hold for every point in $B(I_N)$.
	

	Without loss of generality, assume that $x \in L(I_M)-\bar{V}$ with $\dl(x, U(I_M))>2\delta$.
	We want to show that $x+\vec{\delta}\in I_N$.
	
	Let $x'=\pi_{U(I_M)}(x)$. Observe that $x' > x+2\vec{\delta}$.
	Let $f,g$ be the facets containing $x, x'$ respectively. Choose any $y\in V(f)$ with $y < x$. 
	Such a $y$ exists since $x$ is not a vertex in $f$. 
	Then, we have 
	$$
	I_M \ni y\leq y+2\vec{\delta}\leq x+2\vec{\delta}\in I_M \implies y+2\vec{\delta}\in I_M.
	$$
	By assumption, we have $y+\vec{\delta}\in I_N$. 
	Notice that $y+\vec{\delta}\leq x+\vec{\delta}$.
	
	Let $z\in V(/f/\cap g)$ with $z > x'$. Such a point $z$ always exists by Fact~\ref{fact:vertex_bound}. Then we have 
	$$ 
	I_M\ni x = x+ 2\vec{\delta}-2\vec{\delta} \leq x'-2\vec{\delta}\leq z-2\vec{\delta}\leq z \in I_M   \implies z-2\vec{\delta} \in I_M
	$$
	By assumption, we have $z-\vec{\delta}\in I_N$. Observe that $z-\vec{\delta}\geq x'-\vec{\delta}\geq x+2\vec{\delta}-\vec{\delta}=x+\vec{\delta}$.
	
	Now we have 
	
	$$
	I_N \ni y+\vec{\delta} \leq x+\vec{\delta}\leq z-\vec{\delta}\in I_N \implies x+\vec{\delta}\in I_N
	$$
	
	This completes the proof of (i).
	The proof of (ii) is the same as the one presented for the original proposition.
	
	
	
\end{proof}


In $(n>2)$-parameter case, three things are different from the 2-parameter case from the computational viewpoint: the extended candidate set $\hat{S}$, the discrete set $\bar{V}$ for computing $\delta^*$, and the intersection of intervals in $\Real^n$. We describe the modified algorithm for $n>2$ case below:
\bigskip

\noindent\textbf{Algorithm} {\sc Interleaving ($n>2$)}

(output: $d_I(M, N)$, input: $I_M$ and $I_N$ with $t$ vertices in total)
\begin{enumerate}
	
	\item Compute the candidate set $\hat{S}$ and 
	let $\epsilon$ be the half of the smallest difference between any two numbers in $\hat{S}$.  /* $O(t^2)$ time */
	\item Compute $\delta^*=\max_{x\in \bar{V}}\set{d_I(M|_{\Delta_x}, N|_{\Delta_x})}$; Let $\delta = \delta^*$.  /* $O(t^2)$ time */
	\item Output $\delta$ after a binary search in $S_{\geq \delta^*}$ by following steps /* $O(\log t$) probes */ 
	\begin{itemize}
		\item let $\delta'=\delta+\epsilon$
		\item Compute intersections $I_M \cap I_{N_{\rightarrow\delta'}}$ and $I_N \cap I_{M_{\rightarrow\delta'}}$. /* $O(t^2)$ time */
		\item For each intersection component, check if it is valid or trivializable according to Theorem~\ref{di_criteria}. /* $O(t^2)$ time */
	\end{itemize}
\end{enumerate}

The computation of $\bar{V}$ depends on the intersection $/f/\cap g$ for each pair of facets $f,g\in F(I_M)\cup F(I_N)$. We first compute the projection of $f$ onto the flat $\hat{g}$ of $g$ along the direction $\vec{e}=(1,\cdots, 1)$, denoted as $f_g=/f/\cap \hat{g}$, which is a $(n-1)$-dimensional convex set in $\hat{g}$. Then, we compute the intersection $f_g\cap g\subseteq \hat{g}$. Since we have to do the process for each pair of faces, the entire process takes time $O(t^2)$ where the total number of faces in intervals is $O(t)$.

The computation of $\hat{S}$ depends on the distances from vertices to flats containing the facets. But, since each vertex is contained in a facet, this can be done automatically when we compute $f_g$ in the previous procedure.

In each iteration, the computation of intersection of two intervals requires $O(t^2)$ time. So the total time complexity becomes $O(t^2\log t)$ by taking into account $O(\log t)$ probes in the binary search.

\section{A lower bound on $d_I$}
\label{sec:dimension_function}

In this section we propose a distance between two persistence modules that bounds the interleaving distance from below. This distance is defined for $n$-parameter modules and not necessarily only for 2-parameter modules. It
is based on dimensions of the vectors involved with the two modules and is efficiently computable.

Let $[n]=\set{1,2,\ldots, n}$ be the set of all the integers from $1$ to $n$. Let $\binom{[n]}{k}=\set{s\subseteq [n]: |s|=k}$ be the set of all subset in $[n]$ with cardinality $k$.

\begin{definition}
	%
	
	For a right continuous function $f:\mathbb{R}^n\rightarrow \mathbb{Z}$, 
	define the {\em differential} of $f$ to be $\Delta f: \mathbb{R}^n\rightarrow \mathbb{Z}$ where
	$$
	\Delta f(x) = \sum_{k=0}^n (-1)^k \cdot \sum_{s\in \binom{[n]}{k}} \lim_{\epsilon \to 0_+} f(x-\epsilon \cdot \sum_{i\in s}e_i)
	$$
	Note that for $k=0$, $\sum_{s\in \binom{[n]}{k}} \lim_{\epsilon \to 0_+} f(x-\epsilon \cdot \sum_{i\in s}e_i)=f(x)$. We say $f$ is {\em nice} if the support $supp(\Delta f)$ is finite and $supp(f) \subseteq \set{x\,|\,x\geq \vec{a}}$ for some $a \in \Real$.
\end{definition}


The differential $\Delta f$ is a function recording the change of function values of $f$ at each point, especially at 'jump points'. For $n=1$, $\Delta f(x)=f(x)-\lim_{\epsilon \to 0_+}f(x-\epsilon)$. For $n=2$, which is the case we deal with, we have

$$
\Delta f(x) = f(x) 
- \lim_{\epsilon \to 0_+}f(x-(\epsilon, 0))-\lim_{\epsilon \to 0_+}f(x-(0, \epsilon))
+ \lim_{\epsilon \to 0_+}f(x-(\epsilon, \epsilon)).
$$
See Figure~\ref{dimfunc-fig} and \ref{dm-fig} for illustrations in $1$- and 2-parameter cases respectively.
\begin{proposition}
	\label{prop:nice-func}
	For a nice function $f$, 
	$f(x)=\sum_{y \leq x} \Delta f(y)$ (Proof in Appendix~\ref{app:miss2}).
\end{proposition}


We also define
$\Delta f_{+}=\max\{{\Delta f, 0}\}$, $\Delta f_{-}=\min \{{\Delta f, 0}\}$ and  $f_{\Sigma+}(x) = \sum_{y\leq x}\Delta f_{+}(y)$, $f_{\Sigma-}(x) = \sum_{y\leq x} \Delta f_{-}(y)$. Note that $f_{\Sigma+}\geq 0$, $f_{\Sigma-}\leq 0$, and are both monotonic functions. By definition and property of $\Delta f$, we have $f = f_{\Sigma+} + f_{\Sigma-}$.

\begin{definition}
	For any $\delta>0$, we define the $\delta$-extension of $f$ as $f^{+\delta}=f_+(x+\delta)+f_-(x-\delta)$. Similarly we define the $\delta$-shrinking of $f$ as $f^{-\delta}=f_-(x+\delta)+f_+(x-\delta)$ (see Figure~\ref{dimfunc-fig}).
\end{definition}





Proposition~\ref{prop_delta_further} below follows from the definition.


\begin{figure}[ht!]
	\centering
	\includegraphics[width=0.9\textwidth]{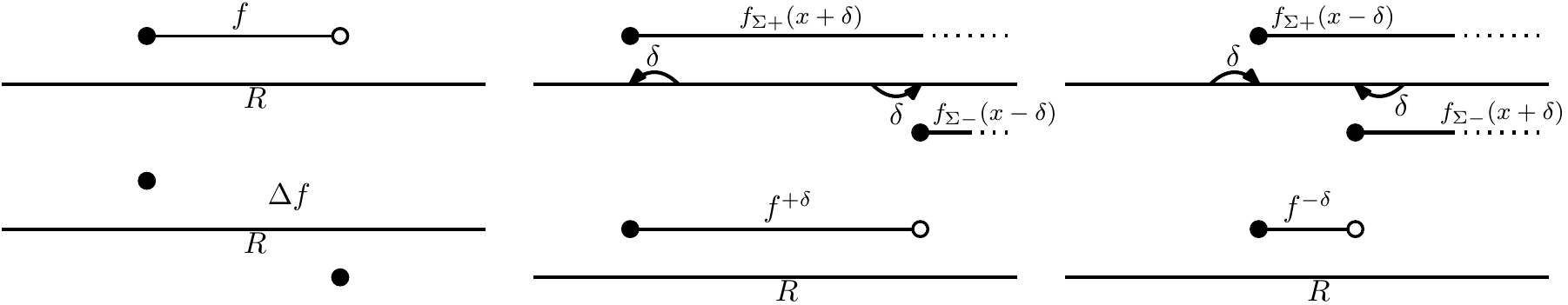}
	\caption{A nice function and its differential (left), its $\delta$-extension (middle), $\delta$-shrinking (right)}
	\label{dimfunc-fig}
\end{figure}

\begin{sloppypar}
	\begin{proposition} \label{prop_delta_further}
		For any $\delta > 0 \in \mathbb{R}$, we have 
		$f^{\pm\delta}(x) = f(x\mp\delta) + \sum_{y\leq x\pm\delta, y \not\leq x\mp\delta} \Delta f_\pm(y)$.
	\end{proposition}
\end{sloppypar}

That is to say, for any $\delta\in \mathbb{R}$, the extended (shrunk) function $f^\delta$ can be computed by adding to $f(x-|\delta|)$ the positive (negative) difference values of $\Delta f$ in $(x-|\delta|, x+|\delta|]$. From this, it follows:

\begin{corollary}
	Given $0\leq \delta \leq \delta' \in \Real$, we have $f^{+\delta}\leq f^{+\delta'}$ and $f^{-\delta}\geq f^{-\delta'}$.
\end{corollary}

\begin{definition}
	For any two nice functions $f,g: \mathbb{R}^n \rightarrow \mathbb{Z}$ and $\delta \geq 0$, we say $f, g$ are within $\delta$-extension, denoted as $f_{\leftarrow \delta \rightarrow} g$,  if $f \leq g^{+\delta}$ and $g \leq f^{+\delta}$. Similarly, we say $f, g$ are within $\delta$-shrinking, denoted as $f_{\rightarrow \delta \leftarrow} g$, if $f \geq g^{-\delta}$ and $g \geq f^{-\delta}$.
\end{definition}

Let $d_+, d_-,d_0$ be defined as follows on the space of all nice real-valued functions on $\mathbb{R}^n$:
\begin{equation*}
d_-(f,g) = \inf_{\delta} \left \{ \delta \mid f_{\rightarrow \delta \leftarrow} g \right\}, d_+(f,g) = \inf_{\delta} \left \{ \delta \mid f_{\leftarrow \delta \rightarrow} g \right\}, d_0(f,g) = \min(d_-, d_+)
\end{equation*}

One can verify that $d_0$ is indeed a distance function. Also, note that when $f,g\geq 0$ (for example, $f,g$ are dimension functions as defined below), we have $d_- \leq d_+$, hence $d_0=d_-$. It seems that the definition of $d_-$ has a similar connotation as the erosion distance defined by Patel~\cite{patel2016generalized} in 1-parameter case. 

\subsection{Dimension distance} \label{sec:stab_of_dimension_func}

\begin{figure}[ht!]
	\centering
	\includegraphics[width=0.5\textwidth]{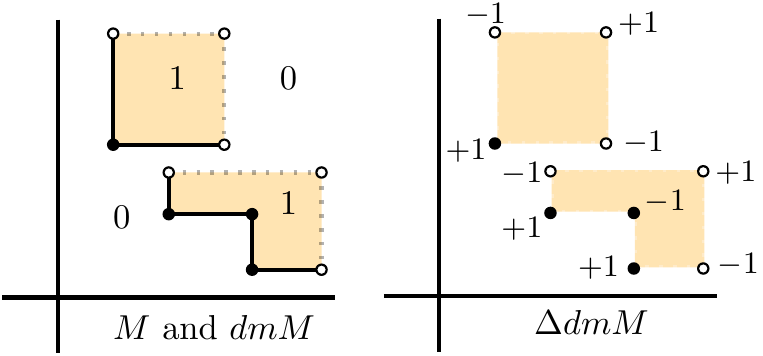}
	\caption{Dimension function (left), its differential is non-zero only at vertices (right)}
	\label{dm-fig}
\end{figure}


Given a persistence module $M$, let the {\em dimension function} $\dm M:\Real^n \rightarrow Z$ be defined as $\dm M(x) = dim(M_x)$. The distance $d_0(dmM,dmN)$ for two modules $M$ and $N$ is called the dimension distance. Our main result in theorem \ref{theorem_stability} is that this distance is stable with respect to the interleaving distance and thus provides a lower bound for it.


\begin{definition}
	A persistence module $M$ is nice if there exists a value $\epsilon_0 \in \mathbb{R}^+$ so that for every $\epsilon < \epsilon_0$, each linear map $\rho^M_{x \rightarrow x+\vec{\epsilon}} : M_{x} \rightarrow M_{x+\vec{\epsilon}}$ is either injective or surjective (or both).
\end{definition}

For example, a persistence module generated by a simplicial filtration defined on a grid with at most one additional simplex being introduced between two adjacent grid points satisfies this nice condition above. 

\begin{sloppypar}
	
	\begin{theorem} \label{theorem_stability}
		For nice persistence modules $M$ and $N$, $d_0(\dm M, \dm N) \leq d_I(M, N)$.
		
	\end{theorem}
	
\end{sloppypar}


\begin{proof}
	Let $d_I(M, N) = \delta$. There exists $\delta$-interleaving, $\phi=\{\phi_x\}, \psi=\{\psi_x\}$ which satisfy both triangular and square commutativity.
	We claim $(\dm M)^{-\delta} \leq	\dm N$ and $(\dm N)^{-\delta} \leq \dm M$.
	
	
	
	
	
	\begin{sloppypar}
		Let $x\in\Real^n$ be any point. 
		By Proposition \ref{prop_delta_further}, we know that
		$(\dm M)^{-\delta}(x) = \dm M(x-\delta) + \sum_{y\leq x+\delta, y\not\leq x-\delta} (\Delta\dm N_-)(y)$.
		If $\dm M(x-\delta) \leq \dm N(x)$, then we get $(\dm M)^{-\delta}(x)\leq \dm M(x-\delta) \leq \dm N(x)$, because $\sum_{y\leq x+\delta, y\not\leq x-\delta} (\Delta\dm N_-)(y) \leq 0$.
	\end{sloppypar}
	
	Now assume $\dm M(x-\delta) > \dm N(x)$. From triangular commutativity, we have 
	$rank(\psi_x \circ \phi_{x-\vec{\delta}}) = rank(\rho^{M}_{x-\vec{\delta}\rightarrow x+\vec{\delta}})$, which gives $\dim(im(\rho^{M}_{x-\vec{\delta}\rightarrow x+\vec{\delta}})) \leq \dim(im(\phi_{x-\vec{\delta}})) \leq \dm N(x)$.
	
	There exists a collection of linear maps $\set{\rho_i:M_{x_i} \rightarrow M_{x_{i+1}}}_{i=0}^k$ such that $\rho^{M}_{x-\vec{\delta}\rightarrow x+\vec{\delta}}=\rho_k\circ\rho_{k-1}\circ\ldots\circ\rho_1\circ\rho_0$ and each $\rho_i$ is either injective or surjective. 
	Let $\Img_i=im(\rho_i\circ\ldots\circ\rho_0)$. 
	Note that $\Img_k=im(\rho^{M}_{x-\vec{\delta}\rightarrow x+\vec{\delta}})$. Let $\epsilon_i=\dim(\Img_i)-\dim(\Img_{i-1})$. 
	Then note that $\epsilon_i=0$ if $\rho_i$ is injective and $\dim(\Img_k)-\dim(M_{x_0})=\sum_{i=1}^k \epsilon_i$. 
	Since $\dim(\Img_k) -\dim(M_{x_0}) < 0$, there exists a collection of $\rho_{i_j}$'s such that $\epsilon_{i_j} < 0$. 
	This means these $\rho_{i_j}$'s are non-isomorphic surjective linear maps with $\dim(M_{x_{i_j}})-\dim(M_{x_{i_j-1}}) < 0$. 
	By definition of $\Delta \dm$, this means that, for each pair $(x_{i_j-1}, x_{i_j})$, 
	there exists a collection $y_1, y_2, \ldots$ such that $y_l \leq x_{i_j}, y_l\not\leq x_{i_j-1}$ and $\sum_l(\Delta\dm M)_{-}(y_l)\leq \epsilon_{i_j}$. 
	All these $y$'s also satisfy that ${y\leq x+\vec{\delta}, y\not\leq x-\vec{\delta}}$. So,
	$$\sum_{\substack {y\leq x+\delta \\ y\not\leq x-\delta}} (\Delta\dm M_-)(y)\leq \sum_j\epsilon_j=\dim(\Img_k)-\dim(M_{x_0})\leq \dim(N_x)-\dim(M_{x-\vec{\delta}}),$$
	which gives $(\dm M)^{-\delta}(x) \leq \dm N(x) $. Similarly, we can show $(\dm N)^{-\delta}(x) \leq \dm M(x)$.
\end{proof}

\subsection{Computation}
For computational purpose, assume that two input persistence modules $M$ and $N$ are finite in that they are functors on the subcategory $\{1,\ldots, k\}^n \subset \Real^n$ and the dimension functions $f:=\dm M$, $g:=\dm N$ have been given as input on an $n$-dimensional $k$-ary grid. 

First, for the dimension functions $f,g$, we compute $\Delta f, \Delta g, \Delta f_\pm, \Delta g_\pm, f_\pm, g_\pm$ in $O(k^2)$ time. 
By Proposition \ref{prop_delta_further}, for any $\delta \in \mathbb{Z}^+$, we can also compute $f^{\pm \delta}, \ g^{\pm \delta}$ in $O(k^2)$ time. Then we can apply the binary search to find the minimal value $\delta$ within a bounded region such that $f,g$ are within $\delta$-extension or $\delta$-shrinking. This takes $O(\log k)$ time. So the entire computation takes $O(k^2 \log k)$ time.

\section{Conclusions}
In this paper, we presented an efficient algorithm to compute the bottleneck distance of two $n$-parameter persistence modules given by indecomposables that may have non-constant complexity. No such algorithm for such case is known. Making the algorithm more efficient will be one of our future goals. Extending the algorithm or its modification to larger classes of modules such as the $n$-parameter modules or exact pfd bi-modules considered in~\cite{cochoy2016decomposition} will be interesting. Here, we assume that indecomposable interval modules have been given as input. Given an $n$-parameter filtration, computing such indecomposables from the resulting persistence module is an important and difficult task. In a recent work, we made a significant progress for this problem, see~\cite{DeyCheng19}.

The assumption of nice modules for dimension distance $d_0$ is needed so that the dimension function, which is a weaker invariant compared to the rank invariants or barcodes in one dimensional case, provides meaningful information without ambiguity. There are cases where the dimension distance can be larger than interleaving distance if the assumption of nice modules is dropped. Of course, one can adjust the definition of dimension distance to incorporate more information so that it remains bounded from above by the interleaving distance.

\section*{Acknowledgments} This research is supported by NSF grants CCF-1526513, 1740761, and DMS-1547357. 
\bibliographystyle{plainurl}
\bibliography{ref}



\section*{Appendix}
\appendix

\section{Missing details in section \ref{sec:esti_dI}}
\label{app:miss1}
{\bf Triangular and square commutative diagrams}.
\[
\begin{tikzcd}
M_x \arrow[rr, "\rho^M_{x \rightarrow x+2\vec{\delta}}"] \arrow[dr, "\phi_x"'] & & M_{x+2\vec{\delta}} & N_x \arrow[rr, "\rho^N_{x \rightarrow x+2\vec{\delta}}"] \arrow[dr, "\psi_x"'] & & N_{x+2\vec{\delta}}\\
& N_{x+\vec\delta} \arrow[ur, "\psi_{x+\vec\delta}"'] & & & M_{x+\delta} \arrow[ur, "\phi_{x+\vec\delta}"']
\end{tikzcd}
\]


\[
\begin{tikzcd}
M_x \arrow[r, "\rho^M_{x \rightarrow y}"] \arrow[d, "\phi_x"']   &   M_{y} \arrow[d, "\phi_y"]  &   M_{x+\vec\delta} \arrow[r, "\rho^M_{x+\vec\delta \rightarrow y+\vec\delta}"]   &   M_{y+\vec\delta} \\
N_{x+\vec\delta} \arrow[r, "\rho^N_{x+\vec\delta \rightarrow y+\vec\delta}"']                        &   N_{y+\vec\delta}                     &   N_{x} \arrow[r, "\rho^N_{x \rightarrow y}"']  \arrow[u, "\psi_x"']                        &   N_{y} \arrow[u, "\psi_y"]
\end{tikzcd}
\]

\noindent
{\bf Proposition~\ref{prop:valid_0} and its proof}.

Let $\{Q_i\}$ be a set of intersection components of $M$ and $N$ with intervals $\{I_{Q_i}\}$. Let $\{\phi_x\}: M\rightarrow N$ be the family of linear maps defined as $\phi_x=\Id$ for all $x \in I_{Q_i}$ and $\phi_x=0$ otherwise. Then $\phi$ is a morphism if and only if every $Q_i$ is $(M,N)$\textit{-valid}.


\begin{proof}
	$\implies$ direction:
	Let $x\in I_{Q_i}$ and $y,z \in \RB^n$ be such that $y \leq x \leq z$. Then,
	\begin{eqnarray*}
		y\in I_M & \implies& \rho^M_{y\rightarrow x} = \Id\\
		&\implies& \phi_x\circ \rho^M_{y\rightarrow x}=\Id=\rho^N_{y\rightarrow x}\circ\phi_y \mbox{ because $\phi$ is a morphism}\\
		&\implies& \phi_y=\Id\\
		&\implies& N_y=\field{k}\\
		&\implies& y\in I_N.
	\end{eqnarray*}
	
	Similarly, we have $z\in I_N\implies z\in I_M$. So, we get $Q_i$ is ($M, N$)-valid.\\
	
	\noindent
	$\Longleftarrow$ direction: We want to show that the square commutativity $\phi_y\circ\rho^M_{x\rightarrow y} = \rho^N_{x\rightarrow y}\circ\phi_x$ holds for any $x\leq y \in \RB^n$ as depicted in the diagram below: 
	
	\[
	\begin{tikzcd}
	M_x \arrow[r, "\rho^M_{x \rightarrow y}"] \arrow[d, "\phi_x"']   &   M_{y} \arrow[d, "\phi_y"]  \\
	N_x \arrow[r, "\rho^N_{x \rightarrow y}"']                        &   N_{y}                     
	\end{tikzcd}
	\]
	
	First, assume that $M$ and $N$ have a single intersection component $Q$ with $I:=I_{Q}$.
	There are several cases.
	
	\textbf{Case 1:} $x, y \in I$: By assumption, every linear map in the square commutative diagram is the identity map. So, it commutes with $\rho$ as required. 
	
	\textbf{Case 2:} $x, y \notin I$: By assumption we have $\phi_x=0, \phi_y=0$. So, it commutes with $\rho$ trivially.  
	
	\textbf{Case 3:} $x\in I$: If $y \in I_N$, then by the assumption that $Q$ is $(M,N)$\textit{-valid}, we have $y\in I_M$. It reduces to case 1. If $y\in I_M\setminus I_N$, we have $\phi_x =\mathbb{1},\phi_y=0, \rho^M_{x\rightarrow y}=\mathbb{1}, \rho^N_{x\rightarrow y}=0$, which imply $\phi_y\circ\rho^M_{x\rightarrow y} =0= \rho^N_{x\rightarrow y}\circ\phi_x$ as required.
	
	\textbf{Case 4:} $y\in I$: If $x \in I_M$, then by assumption that $Q$ is $(M,N)$\textit{-valid}, we have $x\in I_N$. It reduces to case 1.  If $x\in I_N\setminus I_M$, we have $\phi_x =0,\phi_y=\mathbb{1}, \rho^M_{x\rightarrow y}=0, \rho^N_{x\rightarrow y}=\mathbb{1}$, which imply $\phi_y\circ\rho^M_{x\rightarrow y} =0= \rho^N_{x\rightarrow y}\circ\phi_x$ as required.
	
	Now for the case when $M$ and $N$ intersect in a set $\set{Q_i}$ that has more than one element, let $\phi_i$ be the morphism constructed for $Q_i$ only. Then we let $\phi=\set{\phi_x}$ where $\phi_x=\sum_i (\phi_i)_x$. Since each $(\phi_i)_x$ is a scalar function, either $\Id$ or 0 in $\field{k}=\mathbb{Z}/2$, the sum of such morphisms is still a morphism. We can also see that $\phi_x=\mathbb{1}$ for any $x$ in any $I_{Q_i}$ in the set $\set{Q_i}$ and $\phi_x=0$ if $x$ is not in any $I_{Q_i}$. Hence, $\phi$ is a morphism as required. 
	
\end{proof}


\noindent
{\bf Proposition~\ref{closure_interleaving} and its proof.}

$d_I(M, N)=d_I(\overline{M}, \overline{N})$. 

\begin{proof}
	With the triangular inequality of the interleaving distance, the proposition follows straightforwardly from the claim that $d_I(M, \overline{M}) = 0$ which we prove below.
	
	By definition of $\overline{M}$, we have $I_{\overline{M}}=\overline{I_M}$. 
	First, note that each pair of one dimensional slices $M|_{\Delta_x}$ and $\overline{M}|_{\Delta_x}$ are $\delta$-interleaved for any $\delta>0$. That means $\delta^*=0$.
	Let $\delta > 0$ be a small enough number and $I=I_M \cap I_{\overline{M}\rightarrow \delta}$, $J=I_{M\rightarrow\delta} \cap I_{\overline{M}}$.

	We claim that $\forall x\in I, \forall y<x, y\in I_M \implies y\in I_{\overline{M}\rightarrow \delta}$. 
	This is because $\exists w$ such that $y-\vec{\delta} < w < y$ and $ w\in I_{\overline{M}\rightarrow \delta}$. By the property of interval, 
	$$w<y<x \mbox{ and } w,x\in I_{\overline{M}\rightarrow \delta} \implies y\in I_{\overline{M}\rightarrow \delta}.$$ 
	Similarly, we have $\forall x\in I, \forall z>x, z\in I_{\overline{M}\rightarrow \delta} \implies z\in I_M$. Now we construct $\phi = \set{\phi_x:M_x\rightarrow \overline{M}_{x+\delta}}$ by setting $\forall x\in I,\phi_x\equiv \Id$ and $\forall x\notin I, \phi_x \equiv 0$. We define $\psi=\set{\psi_x:\overline{M}\rightarrow M_{x+\delta}}$ in a similar way. Applying similar argument as in the proof of Proposition \ref{prop:valid_0}, one can obtain that these two maps satisfy square commutativity, and hence are morphisms.
	
	Now we claim that $\phi$ and $\psi$ provide a $\delta$-interleaving for each pair of 1-parameter slices $M|_{\Delta_x}$ and $\overline{M}|_{\Delta_x}$, which means they also follow the triangular commutativity. 
	Observe that $\forall x \in I_M, \forall \epsilon > 0, x+2\vec{\epsilon} \in I_M \implies x+\vec{\epsilon}\in I_{\overline M}$. Symmetrically, we have $\forall x \in I_{\overline M}, \forall \epsilon > 0, x+2\vec{\epsilon} \in I_{\overline M} \implies x+\vec{\epsilon}\in I_M$. 
	Now let $\epsilon=\delta$ and consider any nonzero linear map $\rho^M_{x\rightarrow x+2\vec{\delta}}=\Id$ in $M$. Since $x, x+2\vec{\delta}\in I_M\implies x+\vec{\delta}\in{I_{\overline M}}$, we have $x\in I$ and $x+\vec{\delta}\in J$, which imply $\phi_x = \psi_{x+\delta}=\Id$ by our construction of $\phi$ and $\psi$. So, $\forall x$ so that $\rho^M_{x\rightarrow x+2\vec{\delta}}=\Id$, we have $\rho^M_{x\rightarrow x+2\vec{\delta}}=\Id=\psi_{x+\delta}\circ\phi_x$. For those $x$ so that $\rho^M_{x\rightarrow x+2\vec{\delta}}=0$, observe that the commutativity holds trivially. Therefore, $\forall x$, $\rho^M_{x\rightarrow x+2\vec{\delta}}=\psi_{x+\delta}\circ\phi_x$. Symmetrically, we also have the commutativity $\rho^{\overline M}_{x\rightarrow x+2\vec{\delta}}=\phi_{x+\delta}\circ\psi_x$.
	
	Therefore, the morphisms $\phi$ and $\psi$ provide $\delta$-interleaving on the interval modules $M, \overline{M}$. Since this is true for any $\delta > 0$, we get $d_I(M, \overline{M})=0$.
\end{proof}

\noindent
{\bf Proposition~\ref{prop_valid_intersect} and its proof.}

For an intersection component $Q$ of $M$ and $ N$ with interval $I:=I_Q$, the following conditions are equivalent:
\begin{enumerate}[label={(\arabic*)}]
	\item $Q$ is $(M, N)$\textit{-valid}.
	\item $L(I) \subseteq L(I_M)$ and $U(I) \subseteq U(I_N)$.
	\item $VL(I) \subseteq L(I_M)$ and $VU(I) \subseteq U(I_N)$.    
\end{enumerate}

\begin{proof}
	
	\begin{description}
		\item[$(1)\iff (2)$:] Assume (1) is true. Let $x\in L(I)$. 
		For any $y=(y_1, y_2)$ with $y_1< x_1$ and $y_2<x_2$, we have $y\notin I_M$ or $y\notin I_N$ because no such point $y$ can belong to the intersection $I$ as $x$ is on the boundary $L(I)$. Also, by definition of $(M,N)$-validity, $y\notin I_N \implies y\notin I_M$. These two conditions on $y$ imply that $y\notin I_M$.
		Therefore, $x\in L(I_M)$, that is $L(I)\subseteq L(I_M)$.
		
		Similarly, we get $U(I) \subseteq U(I_N)$ proving (1) $\implies$ (2).
		
		Assume (2). 
		Let $x\in I$. For any $y \leq x$, we want to show that  $y\in I_M \implies y\in I_N$, which is equivalent to the condition
		$y\notin I \implies y\notin I_M$
		since $I=I_N\cap I_M$. Observe that $y\notin I\implies y < y'=\pi_{L(I)}(y)$. By assumption that $L(I)\subseteq L(I_M)$, we have $y'\in L(I_M)$, which implies $y< \pi_{L(I_M)}=y'$. So we get $y\notin I_M$. In a similar way, we can get $\forall z\geq x$, $z\notin I \implies z\notin I_N$, or equivalently, $z\in I_N \implies z\in I_M$. Therefore, by definition of $(M,N)$-validity, we obtain (1).


		\item[$(2)\iff (3)$:] $L(I)$ and $U(I)$ are uniquely determined by their vertices.
	\end{description}
\end{proof}

\noindent
{\bf Proposition~\ref{prop:triv} and its proof.}

An intersection component $Q$ is $\delta_{(M,N)}$\textit{-trivializable} if and only if each vertex in $V(I_Q)$ is $\delta_{(M,N)}$\textit{-trivializable}.

\begin{proof}
	Observe that an intersection component $Q$ is $\delta_{(M,N)}$-trivializable if and only if every point in $B(I_Q)$ is $\delta_{(M,N)}$-trivializable. 
	The $\implies$ direction is trivial. For the $\Longleftarrow$ direction,
	observe that, by the definition of $d_{triv}^{(M,N)}$ and Proposition~\ref{prop:distance_bound_2}, we have $\forall x\in B(I_Q)$, $\exists y\in V(I_Q)$, $d_{triv}^{(M,N)}(x)\leq d_{triv}^{(M,N)}(y)$.
	
\end{proof}

Next proposition is used to prove Proposition~\ref{prop:delstar-search}.

\begin{proposition} \label{prop:1d_interleaving_distance}
	Let $M$ and $N$ be two one-parameter interval modules with intervals $I_M=\overline{st}$ and $I_N=\overline{uv}$ respectively. We have
	$\delta \geq d_I(M, N)$ 
	if and only if
	\begin{eqnarray*}
		|s-t|_\infty > 2\delta& \implies& s+\vec{\delta}\in I_N\mbox{ and } t-\vec{\delta}\in I_N, \mbox { and }\\
		|u-v|_\infty > 2\delta&\implies& u+\vec{\delta}\in I_M\mbox{ and } v-\vec{\delta}\in I_M.
	\end{eqnarray*}
\end{proposition}

\begin{proof}
	
	The $\implies$ direction is obvious by the definition of $\delta$-interleaving. For the $\Longleftarrow$ direction,
	we split the premise into two cases.
	
	Case(1): both $|s-t|_\infty \leq 2\delta$ and $|u-v|_\infty \leq 2\delta$ so that the premise holds vacuously. In this case $M, N$ are two bars with length less than or equal to $2\delta$ and one can observe that $d_I(M, N)\leq \delta$.

	Case(2): there is at least one of $|s-t|_\infty$ and $|u-v|_\infty$ 
	which is greater than $2\delta$.
	We want to show that $M$ and $N$ are $\delta$-interleaved by constructing the linear maps $\phi=\set{\phi_x:M_x\rightarrow N_{x+\vec{\delta}}}$ and $\psi=\set{\psi_x:N_x\rightarrow M_{x+\vec{\delta}}}$ explicitly that satisfy both the square commutativity and triangle commutativity.

	
	
	Let $\phi$ and $\psi$ be defined as follows: 
	\begin{equation*}
	\phi_x=   
	\begin{cases}
	\Id, & x\in I_M\cap I_{N\rightarrow\delta}\\
	0,   & otherwise
	\end{cases}
	\qquad
	\psi_x=   
	\begin{cases}
	\Id, & x\in I_N\cap I_{M\rightarrow\delta}\\
	0,   & otherwise\\
	\end{cases}
	\end{equation*}
	By assumption, one can easily verify that for each nonzero linear map $\rho^M_{x\rightarrow x+2\vec{\delta}}$, we have $\rho^M_{x\rightarrow x+2\vec{\delta}}=\Id=\psi_{x+\vec{\delta}}\circ\phi_x$. Similarly, we have $\rho^N_{x\rightarrow x+2\vec{\delta}}=\Id=\phi_{x+\vec{\delta}}\circ\psi_x$. So, $\phi$ and $\psi$ satisfy the triangular commutativity. Now we show that they also satisfy the square commutativity. By Proposition~\ref{prop:valid_0}, it is equivalent to showing that $I_M\cap I_{N\rightarrow\delta}$ is $(M, N_{\rightarrow\delta})$-valid and $I_N\cap I_{M\rightarrow\delta}$ is $(N, M_{\rightarrow\delta})$-valid. We show the first validity, that is, $I_M\cap I_{N\rightarrow\delta}$ is $(M, N_{\rightarrow\delta})$-valid. The second validity can be proved in a similar way.
	
	Observe that, for one dimensional interval modules, $I_M\cap I_{N\rightarrow\delta}$ being $(M, N_{\rightarrow\delta})$-valid is equivalent to saying that $u-\vec{\delta}\leq s$ and $v-\vec{\delta}\leq t$.  
	By assumption of case 2, we know that at least one of $|s-t|_\infty$ and $|u-v|_\infty$ is greater than $2\delta$. 
	Consider the case when $|s-t|_\infty > 2\delta$.
	The other case can be argued similarly.
	By assumption, we have $s+\vec{\delta}\in I_N$.
	This means $u\leq s+\vec{\delta}$, 
	or equivalently, $u-\vec{\delta}\leq s$. Then, the only thing remaining to be shown is that $v-\vec{\delta}\leq t$. Assume on the contrary that $v-\vec{\delta}>t$, which is equivalent to saying $v>t+\vec{\delta}$. Again, by assumption, $t-\vec{\delta}\in I_N$.  This means $u\leq t-\vec{\delta}$, which implies $|v-u|_\infty > |t+\vec{\delta} - (t-\vec{\delta})|_\infty= 2\vec{\delta}$.
	Now by assumption, we have $v-\vec{\delta}\in I_M$, which is contradictory to $v-\vec{\delta}>t$.
\end{proof}

Note that the above proof also works for interval modules with unbounded intervals. For the proposition below, recall that
\begin{eqnarray*}
	D(x)&=&\{\dl(x, L(I_M)), \dl(x, L(I_N)), \dl(x, U(I_M)), \dl(x, U(I_N))\}\\
	S&=& \{d \mid d\in D(x) \mbox{ or $2d\in D(x)$ for some vertex $x\in V(I_M)\cup V(I_N) \}$}.
\end{eqnarray*}


\begin{proposition} \label{prop:distance_bound_2}
	
	Let $M$ and $N$ be two interval modules. 
	Given any point $x\in B(I_M)$ and any $L\in\set{L(I_M),U(I_M),L(I_N),U(I_N)}$ 
	with $x'=\pi_L(x)$ existing, let $d_x=\dl(x, L)$, $\overline{st}$ and $\overline{uv}$ be the two edges containing $x$ and $x'$ respectively. Then there exist (not necessarily distinct) $y, z\in \{s,t,u,v\}$ and $d_x\in D(x), d_y\in D(y)$ such that $d_x\leq d \leq d_y$.

	
	
\end{proposition}

\begin{proof}
	If either $x$ or $x'$ is a vertex, then we just let $y=z=x$ or $x'$ respectively, which provides the conclusion. Now assume neither $x$ nor $x'$ is a vertex, $s\leq t, u\leq v$.
	
	If $x\in B(\RB^2)$, without loss of generality, let $x=(a, +\infty)$.
	
	
	
	If $d=d_\infty(x, x') < +\infty$, then $x'=(a', +\infty)$ for some $a'\in \Real$. If $a'=a$, then $d=0\in S$. If $a'\neq a$, then $x'=u$ or $v$, that is, $x'$ is a vertex, which has been considered before. If $d=d_\infty(x,x')=+\infty$, then $x'=(\pm\infty, +\infty)$. But, in that case, either $s$ or $t$ has the first coordinate different from $x'$, which means either $d_\infty(s, x')=+\infty=d$ or $d_\infty(t, x')=+\infty=d$.

	Now assume $x\in \Real^2$. 
	Let $l_0=\overline{xx'}$ be the line segment with ends $x,x'$. By construction, $l_0$ is contained in the line $\Delta_x$ passing through $x$ that has slope 1.
	For any line segment $l$ in $\Real^2$,
	let $|l|_{\infty}$ be the $d_\infty$ distance between the two end points of $l$. By definition, we know that $x'=\pi_L(x)= \Delta_x \cap L$. So $\dl(x,L)=d_\infty(x, x')=|l_0|$.
	
	Consider the five lines $\Delta_x, \Delta_s, \Delta_{t},\Delta_{u}, \Delta_{v}$ with slope 1. We can order these five lines by their intercepts on the axis of the first coordinate. Note that $\Delta_x$ is ordered third (in the middle) in this sequence. We pick the second and fourth ones in this sequence and observe that they necessarily intersect both edges $\overline{uv}$ and $\overline{st}$. Let $l_1, l_2$ be the line segments on these lines with end points on $\overline{uv}$ and $\overline{st}$. Without loss of generality, we assume $|l_1|_{\infty}\leq |l_2|_{\infty}$. Then we have $|l_1|_{\infty} \leq |l_0|_{\infty} \leq |l_2|_{\infty}$. (See Figure \ref{fig:distance_bound} for an example).
	
	\begin{figure}[ht!]
		\centering
		\includegraphics[scale=0.5]{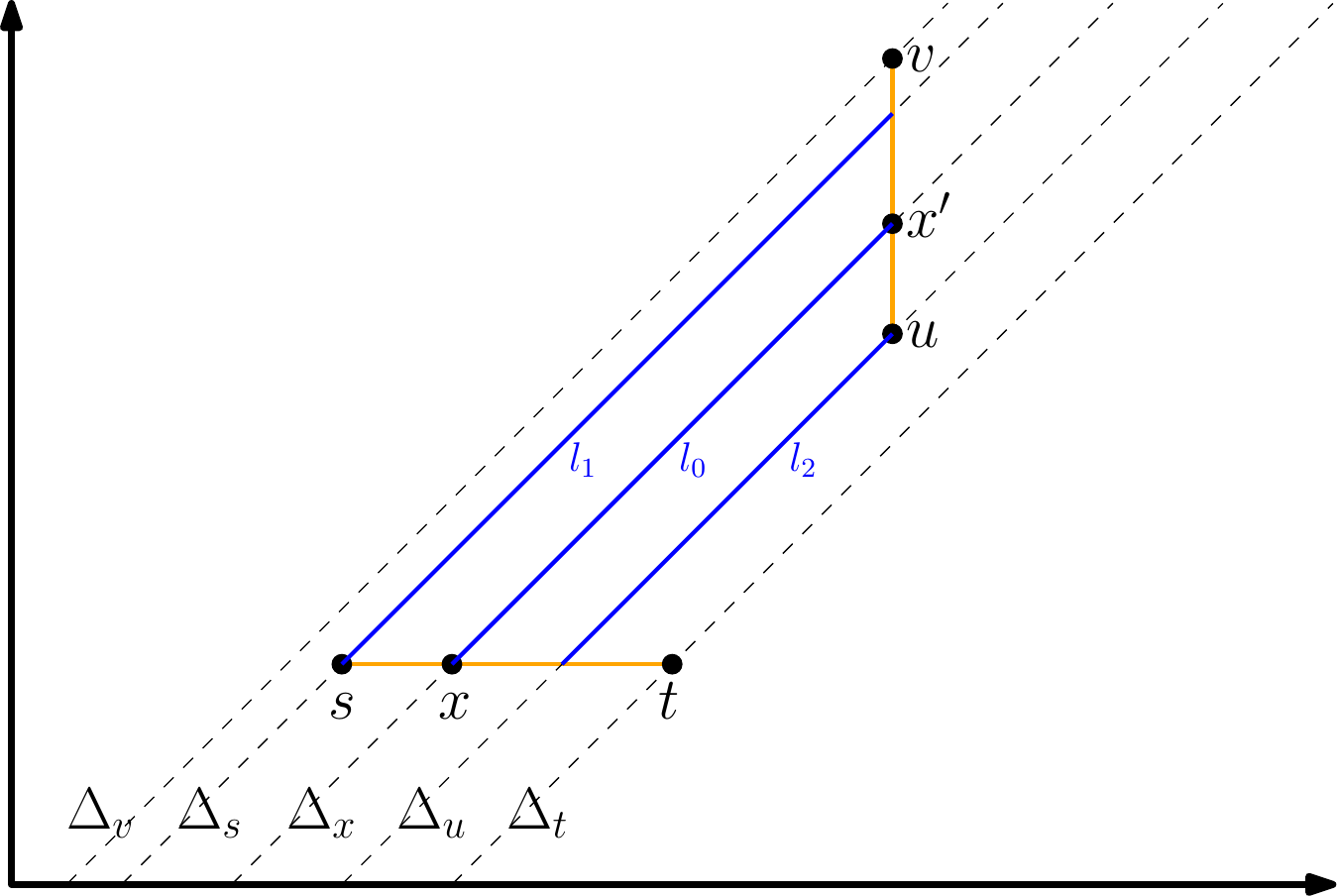}
		\caption{Five diagonal lines (black dotted lines), $\Delta_x, \Delta_s, \Delta_{t},\Delta_{u}, \Delta_{v}$, and three line segments (blue solid line segments), $l_1, l_0, l_2$.}
		\label{fig:distance_bound}
	\end{figure}

	Note that one of the end points of $l_1$ is in the set $\set{s,t,u,v}$, which is a subset of vertices in $V(I_M)\cup V(I_N)$. Let that vertex be $y$. Similarly one of the end points of $l_2$ is a vertex, which we take as $z$.
	We have $|l_1|_{\infty}\in D(y)$ and $|l_2|_{\infty} \in D(z)$ and $d_1=|l_1|_{\infty}\leq d=|l_0|_{\infty}\leq d_2=|l_2|_{\infty}$ for $y, z \in V(I_M)\cup V(I_N)$. This completes the first part of the claim.
	
	
\end{proof}

\begin{corollary}\label{cor:distance_bound}
	If $x$ and $x'$ are on two parallel edges (facets), then $d_y=d_z$. In that case, $d_x=d_y=d_z=d^*\in S$.  
\end{corollary}

\begin{proposition} \label{touching_pt_intersection}
	Let $M$ and $N$ be two interval modules and $d\geq 0$. If there exists an intersection point $x\in B(I_M)\cap B(I_{N\rightarrow d})$ with two parallel facets $f_1\in F(I_M)$ and $f_2\in F(I_{N\rightarrow d})$ both containing $x$, then $d\in S$.

\end{proposition}

\begin{proof}
	Let $\nu:\RB^2 \rightarrow \RB^2$ be the shift function defined as $\nu(x)=x+\vec{d}$. Then $I_N = \nu(I_{N\rightarrow d})$. 
	Let $x'=\nu(x) = x+\vec{d}$ and $f_2'=\nu(f_2)$. Then $f_2'$ and $f_1$ are two parallel facets containing $x'$ and $x$ in $B(I_N)$ and $B(I_M)$ respectively. We know that $f_2'\subseteq L$ for some $L=L(I_N)$ or $U(I_N)$. Then we have $x'=\pi_L(x)$ with $\dl(x,L):=d_\infty(x, x')=d$. By Corollary \ref{cor:distance_bound}, we have $d\in S$. 
\end{proof}

From the above proposition, we get the following corollary. 
\begin{corollary} \label{cor:valid_intersection}
	Let $M$ and $N$ be two interval modules and $d\notin S$. Then, for all intersection points $x\in B(I_M)\cap B(I_{N\rightarrow d})$, any two facets containing $x$ in $B(I_M)$ and $B(I_{N\rightarrow d})$ cannot be parallel,
	that is, $M$ and $N_{\rightarrow d}$ intersect generically. Each intersection component of $M$ and $N_{\rightarrow d}$ results from a transversal intersection.

	
\end{corollary}

\section{Missing proof in section \ref{sec:dimension_function}}
\label{app:miss2}
{\bf Proposition~\ref{prop:nice-func} and its proof}.

For a nice function $f$, 
$f(x)=\sum_{y \leq x} \Delta f(y)$.

\begin{proof}
	For a nice function $f$, we extend $\Delta f$ to be a function $\overline{\Delta f}$ defined on $Pow(\Real^n)$ as $\overline{\Delta f}(U) = \sum_{x\in U} \Delta f(x)$ for any $U \subseteq \Real^n$. Note that $\overline{\Delta f}(\emptyset)=0$ and $\overline{\Delta f}(\{x\}) = \Delta f(x)$. First, we observe the following property of the function $\overline{\Delta f}$:
	\begin{equation} \label{eqn:deltaf_property}
	\overline{\Delta f}(U_1\cup U_2)=\overline{\Delta f}(U_1)+\overline{\Delta f}(U_2) -\overline{\Delta f}(U_1\cap U_2)\tag{$\star$}
	\end{equation}
	
	For any $x=(x_1,\ldots, x_n)\in \Rn$, define $R_x=\set{y:y\leq x}\subseteq\Rn$ and $R_x^{i}=R_x\setminus \set{y:y_i=x_i}=\set{y:y\leq x, y_i\neq x_i}$. For any $k=0,\ldots, n$ and $s\in \binom{[n]}{k}$, let $R_x^{s}=\bigcap_{i\in s} R_x^{i}=\set{y:y\leq x, y_i\neq x_i, \forall i\in s}$. We prove the proposition by induction on $x$. 
	
	Assume it is true for any $y<x$, that is $\forall y< x, f(y)=\sum_{z\leq y} \Delta f(z) = \overline{\Delta f}(R_y)$. Since $R_x=\set{x}\coprod (R_x\setminus \set{x})=\set{x}\coprod \bigcup_i R_x^{i}$, by the property (\ref{eqn:deltaf_property}), we have $\Sigma_{y\leq x} \Delta f(y) =  \overline{\Delta f}( R_x ) = \overline{\Delta f}(\set{x}\coprod \bigcup_i R_x^{i})=\Delta f(x) +\overline{\Delta f}(\bigcup_i R_x^{i})$. By the inclusion–exclusion principle, 
	we have 
	
	\begin{eqnarray*}
		\overline{\Delta f}(\bigcup_i R_x^{i})&=&
		\sum_i\overline{\Delta f}(R_x^{i})-\sum_{ij}\overline{\Delta f}(R_x^{\set{i,j}})+\ldots\\
		&=&(-1)\cdot\sum_{k=1}^n (-1)^{k}\sum_{s\in \binom{[n]}{k}} \overline{\Delta f}(R_x^{s})
	\end{eqnarray*}
	
	Note that by inductive hypothesis, for any $s\in \binom{[n]}{k}$,  $\lim_{\epsilon \to 0_+} f(x-\epsilon \cdot \sum_{i\in s}e_i)
	=\lim_{\epsilon\to 0_+} \overline{\Delta f}(R_{(x-\epsilon\cdot\sum_{i\in s} e_i)})
	=\overline{\Delta f}(\bigcup_{\epsilon > 0} R_{(x-\epsilon\cdot\sum_{i\in s} e_i)})
	=\overline{\Delta f}(R_x^{s})$. Therefore, we have $\overline{\Delta f}(\bigcup_i R_x^{i})=(-1)\cdot\sum_{k=1}^n (-1)^k \cdot \sum_{s\in \binom{[n]}{k}} \lim_{\epsilon \to 0_+} f(x-\epsilon \cdot \sum_{i\in s}e_i)$. By definition of $\Delta f(x)$, we have $f(x)=\Delta f(x)+\overline{\Delta f}(\bigcup_i R_x^{i})=\overline{\Delta f}(R_x)=\sum_{y\leq x} \Delta f(y)$.
\end{proof}




\end{document}